\documentclass[reqno]{amsart}
\usepackage{amsthm, amsmath, amssymb, mathrsfs, enumerate, cite}

\newcounter{mythm}[section]
\newtheorem{theo}[mythm]{Theorem}
\theoremstyle{definition}
\newtheorem{defn}[mythm]{Definition}
\theoremstyle{definition}
\newtheorem{rem}[mythm]{Remark}
\newtheorem{cor}[mythm]{Corollary}
\newtheorem{exa}[mythm]{Example}
\newtheorem{prop}[mythm]{Proposition}
\newtheorem{lem}[mythm]{Lemma}

\numberwithin{equation}{section}

\usepackage{bm}

\newcommand{\qq}{\begin{eqnarray}}
\newcommand{\qqq}{\end{eqnarray}}
\newcommand{\ee}{{\rm e}}

\newcommand{\CO}{{\mathcal O}}

\newcommand{\CT}{{\mathcal T}}

\def\nsection#1{\section{#1}\setcounter{equation}{0}}

\begin{document}
\title{\bf Ergodic properties of a model
for turbulent dispersion of inertial particles}

\author{Krzysztof Gaw\c{e}dzki}
\author{David P. Herzog}
\author{Jan Wehr}
\address{Krzysztof Gaw\c{e}dzki\\ Laboratoire de Physique\\ C.N.R.S., ENS-Lyon,
Universit\'e de Lyon\\ 46 All\'ee d'Italie\\ 69364 
Lyon, France}
\address{David P. Herzog\\ Department of Mathematics\\ The University 
of Arizona\\ 617 N. Santa Rita Ave.\\ P.O. Box 210089\\ Tucson, AZ 85721-0089, USA}
\address{Jan Wehr\\ Department of Mathematics\\ The University of Arizona\\
617 N. Santa Rita Ave.\\ P.O. Box 210089\\ Tucson, AZ 85721-0089, USA}

\begin{abstract}
\noindent We study a simple stochastic differential 
equation that models the  dispersion of close heavy particles 
moving in a turbulent flow. In one and two dimensions, the model 
is closely related to the one-dimensional stationary Schr\"odinger 
equation in a random $\,\delta$-correlated potential. \,The ergodic 
properties of the dispersion process are investigated by proving 
that its generator is hypoelliptic and using control theory.
\end{abstract}

\maketitle
 
%\footnotetext[1]{Laboratoire de Physique, C.N.R.S., ENS-Lyon,
%Universit\'e de Lyon, 46 All\'ee d'Italie, 69364 
%Lyon,\\\hspace*{0.65cm}France}
%\footnotetext[2]{Department of Mathematics, The University of Arizona,
%617 N. Santa Rita Ave. P.O. 
%Box 210089,\\\hspace*{0.65cm}Tucson, AZ 85721-0089, USA}

\nsection{\bf INTRODUCTION}

\noindent Transport by turbulent flows belongs to phenomena
whose understanding is both important for practical applications 
and abounds in intellectual challenges. Unlike the reputedly 
difficult problem of turbulence {\it per se}, turbulent transport 
allows simple modeling that accounts, at least qualitatively, for 
many of its observable features. The simplest of such models study 
transport properties of synthetic random velocity fields with 
presupposed distributions that only vaguely render the 
statistics of realistic turbulent velocities. The advection by velocity 
fields of  quantities like temperature 
or tracer density may be derived from the dynamics of the Lagrangian 
trajectories of fluid elements. In synthetic velocity 
ensembles, such dynamics is described by a random dynamical system.
One of the best studied schemes of this type is the so called 
{\,$\bm{Kraichnan\ model}$\,} based on a Gaussian ensemble of velocities 
decorrelated in time but with long-range spatial correlations 
\cite{Kraich68,FGV}. In this case, the random dynamical system 
that describes the Lagrangian flow is given by stochastic 
differential equations (SDE's). It was successfully studied with the 
standard tools of the theory of random dynamical systems, but it also 
led to non-trivial extensions of that theory \cite{Warwick,LeJanR1,LeJanR2}.  
\vskip 0.1cm

The problem of turbulent transport of matter composed of small but
heavy particles (like water droplets in turbulent atmosphere) may be
also studied by modeling turbulent velocities by a random synthetic 
ensemble, but it requires a modification of the previous approach.
The reason is that heavy particles do not follow Lagrangian trajectories 
due to their inertia. On the other hand, the assumptions of the time 
decorrelation of random velocities may be more realistic for
inertial particles on scales where the typical relaxation time
of particle trajectories (called the \,$\bm{Stokes\ time}$) \,is much longer 
than the typical correlation time of fluid velocities. There have been 
a number of papers that pursued the study of dynamics of inertial 
particles with various simplifying assumptions, see e.g.
\cite{MR,EKR,Piterbarg,FFS,WM1,MW3,DMOW,Horvai,Bec,BCH1,BCH2,FouHorv1,BCHT}.
The primary focus of those studies, combining analytical and numerical 
approaches, was the phenomenon of intermittent clustering of inertial 
particles transported by turbulent flow. A good understanding of that
phenomenon is of crucial importance for practical applications.  
\vskip 0.1cm

The aim of the present article is to show that the simplest among the
models of inertial particles dynamics are amenable to rigorous mathematical 
analysis. More concretely, we study the SDE's that describe the 
pair dispersion of close inertial particles in shortly correlated 
moderately turbulent homogeneous and isotropic $\,d$-dimensional 
velocity fields (not necessarily compressible). Such models were discussed 
in some detail in \cite{Piterbarg,WM1,MW3,MWDWL,Horvai,BCH1}. In particular,
it was noted in \cite{WM1} that the $\,d=1\,$ version of the model
is closely related to the one-dimensional stationary Schr\"odinger equation
with $\,\delta$-correlated potential studied already in the sixties
of the last century \cite{Halperin} as a model for Anderson localization. 
As was stressed in \cite{Horvai}, \,the $\,d=2\,$ model for the inertial 
particle dispersion is also related to the one-dimensional stationary 
Schr\"odinger equation, but this time with $\,\delta$-correlated complex
potential. The models for dispersion were used to extract 
information about the (top) Lyapunov exponent for the inertial particles 
which is a rough measure of the tendency of particles to separate or to 
cluster \cite{WM1,BCH1}. The numerical calculations of the Lyapunov exponents
in two or more dimensional models of particle dispersion presumed certain 
ergodic properties that seemed consistent with results of the simulations. 
Here, we shall establish those properties rigorously by showing the 
hypoellipticity of the generator of the Markov process solving 
the corresponding SDE and by proving the irreducibility of the process
with the help of control theory. For a quick introduction to such, 
by now standard, methods, we refer the reader to \cite{R-B,Friz}. More 
information about the ergodic theory of Markov processes may be found 
in the treatise \cite{MT}.
\vskip 0.1cm

The paper is organized as follows. In Sec.\,\ref{sec:baseq}, we 
present the SDE modeling the inertial particle dispersion. 
In Sec.\,\ref{sec:local}, we recall its relation to models of
one-dimensional Anderson localization. Sec.\,\ref{sec:hypoell}
establishes the hypoelliptic properties of the generator of the dispersion 
process. In Sec.\,\ref{sec:proj}, we introduce the (real-)projective 
version of the dispersion process whose compact space of states may 
be identified with the $\,(2d-1)$-dimensional sphere $\,S^{2d-1}$. 
Sec.\,\ref{sec:ergod} is devoted to proving that the 
dispersion process is controllable. Together with the 
hypoelliptic properties of the generator, this implies that the projectivized version of the process has a
unique invariant probability measure with a smooth strictly
positive density. The analytic expression for such a measure 
may be written down explicitly in $\,d=1\,$ but not in higher
dimensions. The smoothness and strict positivity of its density 
provides, however, in conjunction with the isotropy assumption,
valuable information about the equal-time statistics of the projectivized
dispersion. The isotropy permits to project further the 
projectivized dispersion to the quotient space $\,S^{2d-1}/SO(d)$. 
\,For $\,d=2$, \,this space may be identified with
with the complex projective space $\,\mathbb{P}\mathbb{C}^1
=\mathbb{C}\cup\infty\,$ and the projected process with the 
complex-projectivized dispersion. In fact, the point at infinity may 
be dropped from $\,\mathbb{P}\mathbb{C}^1\,$ since the complex-projectivized 
dispersion process stays in $\,\mathbb{C}\,$ with probability one.
\,For $\,d\geq 3$, \,the quotient space $\,S^{2d-1}/SO(d)\,$ is not smooth
but has an open dense subset that may be identified with the complex
upper-half-plane that the projected process never leaves. \,These
non-explosive behaviors are established in Sec.\,\ref{sec:d>=2}
by constructing a Lyapunov 
function with appropriate properties. Finally, 
in Sec.\,\ref{sec:lyap} and Appendix \ref{app:B}, we demonstrate 
how the established ergodic properties of the projectivized 
dispersion process lead to the formulae for the top Lyapunov exponent 
for inertial particles that were used in the physical literature.
Appendix \ref{app:A} derives a formula, used in the main text, 
that expresses the $\,SO(2d)$-invariant measure on $\,S^{2d-1}\,$ 
in terms of $\,SO(d)\,$ invariants.
\vskip 0.3cm

\noindent{\bf Acknowledgements}. \ J.W. acknowledges the support 
from C.N.R.S. for an extended visit at ENS-Lyon during which this work 
was started.  D.P.H. and J.W. acknowledge partial support under NSF Grant DMS 0623941.

\nsection{\bf BASIC EQUATIONS}
\label{sec:baseq}

\noindent The motion in a turbulent flow of a small body of large density, 
called below an {\,\bf inertial particle}, \,is well described by the equation 
\cite{MR,WM1,MW3,Bec,BCH1} 
\qq
\ddot{\bm r}\ =\ -\frac{_1}{^{\tau}}\,\big(\dot{\bm r}
-{\bm u}(t,{\bm r})\big)\,,
\label{Newteq}
\qqq
where ${\bm r}(t)$ is the position of the particle at time $\,t\,$
and $\,{\bm u}(t,{\bm r})\,$ is the fluid velocity field. \,Relation
(\ref{Newteq}) is the Newton equation with the particle acceleration 
determined by a viscous friction force proportional to the relative 
velocity of the particle with respect to the fluid. 
Constant $\,\tau\,$ is the Stokes time. 
%\,The second order equation (\ref{Newteq}) may be rewritten in the first 
%order form: 
%\qq
%\dot{\bm r}\ =\ {\bm v}\,,\qquad\dot{\bm v}\ =\ -\frac{_1}{^\tau}
%\big({\bm v}-{\bm u}(t,{\bm r})\big)\,.
%\label{Hameq}
%\qqq
Much of the characteristic features of the distribution of non-interacting 
inertial particles moving in the flow according to Eqs.\,(\ref{Newteq}) 
%or (\ref{Hameq}) 
is determined by the dynamics of the  
separation $\,\delta{\bm r}(t) \equiv{\bm\rho}(t)$, \,called 
\,$\bm{particle\ dispersion}$, \,of very close 
trajectories. \,In a moderately turbulent flow, the particle dispersion 
evolves according to the linearized equation:
\qq
\ddot{\bm\rho}\ =\ -\frac{_1}{^{\tau}}\,\big(\dot{\bm\rho}-({\bm\rho}
\cdot{\bm\nabla}){\bm u}(t,{\bm r}(t))\big)
\label{inteq0}
\qqq
or, \,in the first-order form:
\qq
\dot{\bm\rho}\ =\ \frac{_1}{^\tau}{\bm\chi}\,,\qquad\dot{\bm\chi}\ 
=\ -\frac{_1}{^\tau}
{\bm\chi}\,+\,({\bm\rho}\cdot{\bm\nabla}){\bm u}(t,{\bm r}(t))\,.
\label{inteq1}
\qqq
For sufficiently heavy particles, the correlation time of 
$\,({\bm\nabla}{\bm u})(t,{\bm r}(t))\,$ is short with respect to the 
Stokes time $\,\tau\,$ and one may set in good approximation \cite{BCH1}
\qq
\nabla_ju^i(t,{\bm r}(t))\,dt\ =\ dS^i_j(t)
\label{mnoise}
\qqq
where $\,dS(t)\,$ is a matrix-valued white noise with the isotropic covariance
\qq
\big\langle dS^i_j(t)\,dS^k_l(t')\big\rangle=D^{ik}_{jl}\,\delta(t-t')\,dt
\,dt'\,,
\quad\ D^{ik}_{jl}=A\,\delta^{ik}\delta_{jl}+
B(\delta^i_j\delta^k_l+\delta^i_l\delta^k_j)\,.
\label{icov}
\qqq
Positivity of the covariance requires that
\qq
A\geq|B|\,,\qquad A+(d+1)B\geq0\,.
\qqq
Incompressibility implies that $\,A+(d+1)B=0$, \,but we shall not 
impose it, in general. \,We shall only assume that $\,A+2B>0\,$ for $\,d=1\,$ 
and that $\,A>0\,$ for $\,d\geq2$. \,After the substitution of 
(\ref{mnoise}),  \,Eq.\,(\ref{inteq0}) becomes the linear SDE
\qq
\ddot{\bm\rho}\ =\ -\frac{_1}{^{\tau}}\,\dot{\bm\rho}\,+\,\frac{_1}{^\tau}\,
\frac{_{dS(t)}}{^{dt}}{\bm\rho}
\label{inteq0S}
\qqq
that may be written in the first order form in a more standard notation
employing differentials as
\qq
d\left(\begin{matrix}{\bm\rho}\cr{\bm\chi}\end{matrix}\right)\ 
=\ \left(\begin{matrix}0&\hspace{-0.2cm}\frac{1}{\tau}dt\cr dS(t)&
\hspace{-0.2cm}-\frac{1}{\tau}dt\end{matrix}\right)
\left(\begin{matrix}{\bm\rho}\cr{\bm\chi}\end{matrix}\right).
\label{inteq2}
\qqq
We shall interpret the latter SDE using the It\^{o} convention, but
the Stratonovich convention would lead to the same process.
The solution of Eq.\,(\ref{inteq2}) exists with probability 1 for all 
times and has the form
\qq
\left(\begin{matrix}{\bm\rho}(t)\cr{\bm\chi}(t)\end{matrix}\right)\ =\ 
{\CT}\,\exp\Big[\int\limits_0^t\left(\begin{matrix}0&\hspace{-0.2cm}
\frac{1}{\tau}ds\cr
dS(s)&\hspace{-0.2cm}-\frac{1}
{\tau}ds\end{matrix}\right)\Big]\,
\left(\begin{matrix}{\bm\rho}(0)\cr{\bm\chi}(0)\end{matrix}\right),
\label{solut}
\qqq  
where the time ordered exponential may be defined as the sum of its
Wiener chaos decomposition
\qq
&&{\CT}\,\exp\Big[\int\limits_0^t\left(\begin{matrix}0&\hspace{-0.2cm}
\frac{1}{\tau}ds\cr
dS(s)&\hspace{-0.2cm}-\frac{1}{\tau}ds\end{matrix}\right)\Big]\ \ =\ \ 
\sum\limits_{n=0}^\infty\hspace{0.1cm}
\int\limits_{0<s_1<\cdots<s_n<t}\hspace{-0.2cm}
\ee^{(t-s_n)\Big(\begin{matrix}{}_0&\hspace{-0.1cm}_{\frac{1}{\tau}}\cr{}^0&
\hspace{-0.1cm}^{-{\frac{1}{\tau}}}\end{matrix}\Big)}\cr
&&\cdot\,\left(\begin{matrix}0&\hspace{-0.2cm}0\cr dS(s_n)&\hspace{-0.2cm}0
\end{matrix}\right)
\,\ee^{(s_n-s_{n-1})\Big(\begin{matrix}{}_0&\hspace{-0.15cm}_{\frac{1}{\tau}}\cr^0&
\hspace{-0.15cm}^{-\frac{1}{\tau}}\end{matrix}\Big)}\cdots\,
\left(\begin{matrix}0&\hspace{-0.2cm}0\cr dS(s_1)&\hspace{-0.2cm}0
\end{matrix}\right)
\,\ee^{s_1\Big(\begin{matrix}{}_0&\hspace{-0.1cm}_{\frac{1}{\tau}}\cr^0&\hspace{-0.1cm}
^{-\frac{1}{\tau}}\end{matrix}\Big)}ds_1\hspace{-0.07cm}\dots ds_n\qquad
\qqq
that converges in the $\,L^2$-norm for functionals of the white noise 
$\,dS(t)$. \,The resulting stochastic process $\,({\bm\rho}(t),{\bm\chi}(t))
\equiv{\bm p}(t)\,$ is Markov and has generator 
\qq
L\ =\ {\frac{_1}{^\tau}}\big({\bm\chi}\cdot{\bm\nabla}_{\bm\rho}\,-\,
{\bm\chi}\cdot{\bm\nabla}_{\bm\chi}\big)\,+\,\frac{_1}{^{2}}\sum\limits_{i,j,k,l}
\rho^j\rho^l\,D^{ik}_{jl}\,
\nabla_{\chi^i}\nabla_{\chi^k}\,.
\label{genL}
\qqq
In other words, for smooth functions $\,f$,
\qq
\frac{d}{dt}\Big\langle f({\bm p}(t))\Big\rangle\,=\,
\Big\langle (Lf)({\bm p}(t))\Big\rangle
\qqq
for $\,\big\langle\,\cdots\,\big\rangle\,$ denoting the expectation.
\vskip 0.1cm

For the process $\,{\bm p}(t)\,$ given by 
Eq.\,(\ref{solut}), $\,{\bm p}(t)=0\,$ for all $\,t\geq0\,$  if 
$\,{\bm p}(0)=0$.  \,On the other hand, if $\,{\bm p}(0)
\not=0\,$ then $\,{\bm p}(t)\not=0\,$ with probability 
1 for all $\,t\geq0\,$ so that we may restrict the space of states
of the Markov process $\,{\bm p}(t)\,$ to $\,\mathbb{R}^{2d}\setminus\{0\}
\equiv\mathbb{R}^{2d}_{\not=0}$.

\nsection{\bf RELATION TO ONE-DIMENSIONAL LOCALIZATION}
\label{sec:local}

\noindent In the lowest dimensions, there is a simple relation between 
the stochastic process $\,{\bm p}(t)=({\bm\rho}(t),{\bm\chi}(t))\,$ 
constructed above and simple models of Anderson localization in one space 
dimension. Let us set 
\qq
{\bm\psi}(t)\ =\ \ee^{\frac{t}{2\tau}}{\bm\rho}(t)
\label{psirho}
\qqq
exponentially blowing up the long-time values of $\,{\bm\rho}(t)$.
\,Eq.\,(\ref{inteq0S}) implies that
\qq
-\ddot{\bm\psi}\,+\,\frac{_1}{^\tau}\,\frac{_{dS(t)}}{^{dt}}\,{\bm\psi}\ =\ 
-\frac{_1}{^{4\tau^2}}\,{\bm\psi}\,,
\label{loceq}
\qqq
or, \,in the first order form,
\qq
d\left(\begin{matrix}{\bm\psi}\cr{\bm\xi}\end{matrix}\right)\ 
=\ \left(\begin{matrix}0&{\frac{1}{\tau}}dt\cr dS(t)
+\frac{1}{4\tau}dt&0\end{matrix}\right)
\left(\begin{matrix}{\bm\psi}\cr{\bm\xi}\end{matrix}\right)\,.
\label{inteq3}
\qqq
Similarly as before, the above SDE defines a Markov process. 
%with generator
%\qq
%\Lambda\ =\ {\frac{_1}{^\tau}}\big({\bm\xi}\cdot{\bm\nabla}_{\bm\psi}\,
%+\,\frac{_1}{^{4}}
%{\bm\psi}\cdot{\bm\nabla}_{\bm\xi}\big)\,+\,\frac{_1}{^{2}}
%\sum\limits_{i,j,k,l}\psi^j
%\psi^lD^{ik}_{jl}\nabla_{\xi^i}\nabla_{\xi^k}\,.
%\qqq
Clearly,
\qq
({\bm\psi}(t),{\bm\xi}(t))\,=\,\ee^{\frac{t}{2\tau}}\big({\bm\rho}(t),
{\bm\chi}(t)+\frac{_1}{^{2}}{\bm\rho}(t)\big)\,.
\qqq 
\vskip 0.2cm

Viewing $\,t\,$ as the one-dimensional spatial coordinate, 
\,Eq.\,(\ref{loceq}) takes the form of the vector-like stationary 
Schr\"{o}dinger equation
\qq
-\frac{_{d^2}}{^{dt^2}}{\bm\psi}\,+\,V(t)\,{\bm\psi}\ =\ 
E\,{\bm\psi}\,,
\label{local1d}
\qqq
where $\,V(t)=\frac{1}{\tau}\,\frac{dS(t)}{dt}\,$ plays the role of
the random matrix-valued white-noise potential and $\,E=-\frac{1}{4\tau^2}\,$
of the (negative) energy. \,In particular, in $\,d=1$, \,$\,{\bm\psi}(t)\,$
is a real scalar function and so is the $\,\delta$-correlated potential 
\qq
V(t)\,=\,\frac{1}{\tau}\sqrt{A+2B}\,\frac{d\beta(t)}{dt}\,,
\qqq
where $\,\beta(t)\,$ is the Brownian motion. The scalar version 
of Eq.\,(\ref{local1d}) was studied in \cite{Halperin} as a model of 
one-dimensional Anderson localization, see also \cite{LifGredPas}. 
\,In $\,d=2$, \,interpreting $\,{\bm\psi}\,$ as a complex
number $\,\psi^1+i\psi^2$, \,one may replace the matrix valued
$\,\delta$-correlated potential $\,V(t)\,$ in the SDE 
(\ref{local1d}) with the complex valued one 
\qq
V(t)\,=\,\frac{_1}{^\tau}
\Big(\sqrt{A+2B}\,\frac{{d\beta^1(t)}}{{dt}}
+i\sqrt{A}\,\frac{{d\beta^2(t)}}{{dt}}\Big)
\qqq
where $\,\beta^1(t),\ \beta^2(t)$ are two independent Brownian motions
(the two realizations of $\,V(t)\,$ lead to the Markov processes with the same
generator and, consequently, with the same law). \,Consequently, as stressed
in \cite{Horvai}, \,Eq.\,(\ref{local1d}) in $d=2$ may be viewed as a model 
of localization for a one-dimensional non-hermitian random Sch\"odinger 
operator of the type not studied before.

\nsection{\bf HYPOELLIPTIC PROPERTIES OF THE GENERATOR}
\label{sec:hypoell}

\noindent The generator (\ref{genL}) of the particle dispersion process 
$\,{\bm p}(t)=({\bm\rho}(t),{\bm\chi}(t))\,$ has certain non-degeneracy 
properties
which imply smoothness of the transition probabilities. 
\,Let us start with the following fact about the covariance (\ref{icov})
of the matrix-valued white noise $\,dS(t)\,$: 
\vskip 0.3cm

\begin{lem} \ \ We have
\qq
D^{ik}_{jl}\ =\ \sum\limits_{m,n=1}^d\big(E\delta^i_j\delta^m_n
+F\delta^{im}\delta_{jn}+G\delta^i_n\delta_j^m\big)
\big(E\delta^k_l\delta^m_n+F\delta^{km}\delta_{ln}+G\delta^k_n\delta_l^m\big)
\label{forD}
\qqq
for 
\vskip -0.8cm
\qq
&&E\,=\,d^{-1}\big(-\sqrt{A+B}\,+\,\sqrt{A+(d+1)B}\,\big)\,,\cr
&&F\,=\,\frac{_1}{^2}(\sqrt{A+B}+\sqrt{A-B})\,,\cr
&&G\,=\,\frac{_1}{^2}(\sqrt{A+B}-\sqrt{A-B})\,.
\qqq
\end{lem}

\begin{proof}\ \ The right hand side of Eq.\,(\ref{forD}) is:
\qq
(E^2d+2EF+2EG)\delta^i_j\delta^k_l\,+\,(F^2+G^2)\delta^{ik}\delta_{jl}
\,+\,2FG\delta^i_l\delta^k_j\,.
\qqq
Hence, in order to satisfy Eq.\,(\ref{forD}), we must have
\qq
A=F^2+G^2\,,\qquad B=E^2d+2E(F+G)=2FG\,.
\qqq
The assumed values of $E,\,F,\,G\,$ solve these equations.
\end{proof}

\noindent Define the vector fields
\qq
X_0\,\,&=&{\frac{_1}{^\tau}}\big({\bm\chi}\cdot{\bm\nabla}_{\bm\rho}\,
-\,{\bm\chi}\cdot{\bm\nabla}_{\bm\chi}\big)\,,\cr\cr
X^m_n&=&\sum\limits_{i,j}\rho^j\big(E\delta^i_j\delta^m_n
+F\delta^{im}\delta_{jn}+G\delta^i_n\delta_j^m\big)\,\nabla_{\chi^i}\,,\cr\cr
Y^m_n&=&-\sum\limits_{i,j}\rho^j\big(E\delta^i_j\delta^m_n
+F\delta^{im}\delta_{jn}+G\delta^i_n\delta_j^m\big)\,\nabla_{\rho^i}\cr
&&+\,\sum\limits_{i,j}(\chi^j+\rho^j)
\big(E\delta^i_j\delta^m_n
+F\delta^{im}\delta_{jn}+G\delta^i_n\delta_j^m\big)\,\nabla_{\chi^i}\,
=\,\tau[X_0,X^m_n]\,,\cr\cr
Z^m_n&=&-\sum\limits_{i,j}(2\chi^j+\rho^j)
\big(E\delta^i_j\delta^m_n+F\delta^{im}\delta_{jn}+G\delta^i_n
\delta_j^m\big)\,\nabla_{\rho^i}\cr
&&+\,\sum\limits_{i,j}(\chi^j+\rho^j)
\big(E\delta^i_j\delta^m_n+F\delta^{im}\delta_{jn}
+G\delta^i_n\delta_j^m\big)\,\nabla_{\chi^i}\,=\,\tau[X_0,Y^m_n]\,.
\qqq
Using Eq.\,(\ref{forD}), \,one infers that Eq.\,(\ref{genL}) 
giving the generator $\,L\,$ 
may be rewritten in the form:
\qq
L\ =\ X_0\,+\,\frac{_1}{^2}\sum\limits_{m,n}(X^m_n)^2\,.
\label{LX}
\qqq

\begin{rem}
\label{rem:eqSDE}
The process $\,{\bm p}(t)\,$ may be equivalently obtained from
the SDE
\qq
d{\bm p}\ =\ {\bm X}_0({\bm p})\,dt\,
+\sum\limits_{m,n=1}^d{\bm X}^m_n({\bm p})\,d\beta^n_m(t)\,,
\qqq
where $\,\beta^n_m(t)\,$ are independent Brownian motions. Here, we adopt the Stratonovich convention and hence the generator corresponding to that equation has the form (\ref{LX})
so that the latter SDE leads to a process with the same law as
$\,{\bm p}(t)$.  The convention of the stochastic integral we choose, however, is insignificant as the process $\bm{\rho}(t)$ is of bounded variation.    
\end{rem}

In order to establish hypoelliptic properties of $\,L$, \,we shall use
the following non-degeneracy relation satisfied by the vector fields 
$\,X^m_n,\ Y_n^m\,$ and $\,Z_n^m\,$:

\begin{prop}
\label{prop:1}\ \ Suppose that 
$\,{\bm p}=({\bm\rho},{\bm\chi})
\not=0$. \,Then the vectors
\qq
X^m_n({\bm p})\,,\quad Y^m_n({\bm p})\,,
\quad Z^m_n({\bm p})\qquad{\rm with}\quad
{m,n=1,\dots,d}
\label{span}
\qqq
span the $2d$-dimensional space.
\end{prop}

\begin{proof}\ \ First suppose that $\,{\bm\rho}={0}\,$
so that $\,{\bm\chi}\not={0}$. \,We have
\qq
Y^m_n({0},{\bm\chi})\ =\ \sum\limits_{i,j}
\chi^j\big(E\delta^i_j\delta^m_n
+F\delta^{im}\delta_{jn}+G\delta^i_n\delta_j^m\big)\,\nabla_{\chi^i}\,.
\qqq
Let $\,{\bm\phi}\in\mathbb{R}^d$. \,Then
\qq
\sum\limits_{m,n}(\alpha\,\chi^n\chi^m+\beta\,\chi^n\phi^m)\,
\nonumber Y^m_n({0},{\bm\chi})&=&\ [E(\alpha{\bm\chi}^2+\beta{\bm\chi}\cdot{\bm\phi})
+F\alpha\,{\bm\chi}^2+G(\alpha\,{\bm\chi}^2+\beta\,{\bm\chi}\cdot{\bm\phi})]\,
{\bm\chi}\cdot{\bm\nabla}_{\bm\chi}\,\\
&\,& +\,F\beta\,{\bm\chi}^2\,
{\bm\phi}\cdot{\bm\nabla}_{\bm\chi}\,.\quad
\qqq
Setting
\qq
\alpha=-\,\frac{E+G}{(E+F+G)F}\,\frac{{\bm\chi}\cdot{\bm\phi}}
{({\bm\chi}^2)^2}\,,\qquad \beta=\frac{1}{F\,{\bm \chi}^2}
\qqq
(note that $\,F>0\,$ and $\,E+F+G>0$), \,we obtain
\qq
\sum\limits_{m,n}(\alpha\,\chi^n\chi^m+\beta\,\chi^n\phi^m)
\,Y^m_n({0},{\bm\chi})\ =\ {\bm\phi}\cdot{\bm\nabla}_{\bm\chi}\,.
\qqq
Hence the vector $\,{\bm\phi}\cdot{\bm\nabla}_{\bm\chi}\,$ is in the span
of (\ref{span}) for arbitrary $\,{\bm\phi}$. \,We have still to show that 
an arbitrary vector $\,{\bm\sigma}\cdot{\bm\nabla}_{\bm\rho}\,$ is in
that span. \,To this aim note that
\qq
Z^m_n({0},{\bm\chi})\ =\ -2\sum\limits_{i,j}
\chi^j\big(E\delta^i_j
\delta^m_n+F\delta^{im}\delta_{jn}+G\delta^i_n\delta_j^m\big)\,
\nabla_{\rho^i}\,+\,\sum\limits_{i}(...)^i\nabla_{\chi^i}\,.\ 
\qqq
Proceeding as before, we show that an appropriate combination
of $\,Z^m_n({0},{\bm\chi})\,$ gives the vector
\qq
{\bm\sigma}\cdot{\bm\nabla}_{\bm\rho}\,+\,\sum\limits_{i}
(...)^i\nabla_{\chi^i}
\qqq
from which the term $\,\sum(...)^i\nabla_{\chi^i}\,$ may be removed by 
subtracting
an appropriate combination of $\,Y^m_n({0},{\bm\chi})$.
That ends the proof of the claim of \,Proposition \ref{prop:1} 
for $\,{\bm\rho}={0}$.
\vskip 0.1cm

Suppose now that $\,{\bm\rho}\not={0}$. \,Proceeding as before, we see
that arbitrary vector $\,{\bm\phi}\cdot{\bm\nabla}_{\bm\chi}\,$ may be
obtained by taking an appropriate combination of the vectors
$\,X^m_n({\bm\rho},{\bm\chi})$. \,Similarly, arbitrary vector 
$\,{\bm\sigma}\cdot{\bm\nabla}_{\bm\rho}\,$ may be obtained as an appropriate
combination of the vectors $\,Y^m_n({\bm\rho},{\bm\chi})\,$ and
$\,X^m_n({\bm\rho},{\bm\chi})$. \,This completes the proof of \,Proposition 
\ref{prop:1}.
\end{proof}

\noindent The representation (\ref{LX}) and Proposition \ref{prop:1} imply,
in virtue of H\"{o}rmander's theory \cite{Hoerm,Nualart,Norr,Friz}, 
the following result:
\vskip 0.3cm

\begin{cor} \ \ The operators $\,L,\ L^\dagger,\ \partial_t-L,\ 
\partial_t-L^\dagger\,$ and $\,2\partial_t-L\otimes 1-1\otimes L^\dagger\,$
are hypoelliptic\footnote{A differential
operator  $\,D\,$ on a domain $\,\Omega\,$ is hypoelliptic
if for all distributions $\,f,g\,$ such 
that\\\hspace*{0.65cm}$Df=g$, \,smoothness of $\,g\,$ on an open subset 
$\,U\subset\Omega\,$  implies smoothness of $\,f\,$ on $\,U$.} on 
$\,\mathbb{R}^{2d}_{\not=0}$, $\,\mathbb{R}_+\times\mathbb{R}^{2d}_{\not=0}\,$
and $\,\mathbb{R}_+\times\mathbb{R}^{2d}_{\not=0}\times
\mathbb{R}^{2d}_{\not=0}$,
\,respectively.
\end{cor}

\noindent In particular, the hypoellipticity 
of $\,2\partial_t-L\otimes 1-1\otimes L^\dagger\,$
implies that the transition probabilities of the dispersion process 
$\,{\bm p}(t)$,
\qq
P_t({\bm p}_0,d{\bm p})\ =\ P_t({\bm p}_0,
{\bm p})\,d{\bm p}\,, 
\qqq
have densities (annihilated by 
$\,2\partial_t-L\otimes 1-1\otimes L^\dagger$) 
\,that are smooth functions of 
$\,(t,{\bm p}_0,{\bm p})\,$
for $\,t>0\,$ and away from the origin in $\,\mathbb{R}^{2d}$.

\nsection{{\bf CONTROL THEORY AND IRREDUCIBILITY}} 
\label{sec:ergod}

\noindent The additional important property of the process $\,{\bm p}(t)\,$
restricted to $\,\mathbb{R}^{2d}_{\not=0}\,$
is its $\,\bm{irreducibility}\,$ \,assured by the strict positivity
the smooth transition probability densities 
$\,P_t({\bm p}_0,{\bm p})\,$ for all $\,t>0\,$ and 
$\,{\bm p}_0\not=0\not={\bm p}$. \,The latter property 
results, \,according to Stroock-Varadhan's 
Support Theorem \cite{SV}, see also \cite{R-B}, \,from the 
$\,\bm{controllability}\,$ of the process 
$\,{\bm p}(t)\,$ on $\,\mathbb{R}^{2d}_{\not=0}\,$ that is established 
in the following 
\vskip 0.3cm

\begin{prop}\label{prop:2}\ \ For every $\,T>0\,$ and
$\,{\bm p}_0\not=0\not={\bm p}_1\,$ there 
exists a piecewise smooth curve $\,[0,T]\ni t\mapsto(u^n_m(t))
\in\mathbb{R}^{d^2}\,$ such that the solution of the ODE
\qq
\dot{\bm p}\,=\,X_0({\bm p})\,+\,
\sum\limits_{m,n}u^n_m(t)\,X^m_n({\bm p})
\label{contr}
\qqq
with the initial condition $\,{\bm p}(0)={\bm p}_0\,$
satisfies $\,{\bm p}(T)={\bm p}_1$. 
\end{prop}

\begin{proof}\ \ First suppose that $\,{\bm\rho}_0\not={0}\not=
{\bm\rho}_1$. \,Let $\,[0,T]\ni t\mapsto{\bm\rho}(t)\,$ be any curve such that
\qq
&{\bm\rho}(0)={\bm\rho}_0\,,\qquad &\tau\dot{\bm\rho}(0)={\bm\chi}_0\,,\cr
&{\bm\rho}(T)={\bm\rho}_1\,,\qquad &\tau\dot{\bm\rho}(T)={\bm\chi}_1
\qqq
and such that $\,{\bm\rho}(t)\not={0}\,$ for all $\,t\in[0,T]$.
\ Set $\,{\bm\chi}(t)=\tau\dot{\bm\rho}(t)$. \,Let 
\qq
{\bm\phi}(t)\ =\ \tau\dot{\bm\chi}(t)+{\bm\chi}(t)\,.
\qqq
Then the formula
\qq
u^n_m\ =\ \frac{_1}{^\tau}(\alpha\,\rho^n\rho^m\,+\,\beta\,\rho^n\phi^m\big)\,,
\qqq
where now
\qq
\alpha=-\,\frac{E+G}{(E+F+G)F}\,\frac{{\bm\rho}\cdot{\bm\phi}}{({\bm\rho}^2)^2}\,,
\qquad \beta=\frac{1}{F\,{\bm\rho}^2}\,,
\qqq
defines smooth control functions $\,[0,T]\ni t\mapsto(u^n_m(t))\,$ such that
Eq.\,(\ref{contr}) holds.
\vskip 0.1cm

Now suppose that $\,{\bm\rho}_0={0}\not={\bm\rho}_1$. \,Choose $\,0<\epsilon<
\frac{_1}{^2}T\,$
and for $\,0\leq t\leq\epsilon$, \,set
\qq
{\bm\rho}(t)\ =\ \big(1-\ee^{-\frac{t}{\tau}}\big){\bm\chi}_0
\qqq
and $\,{\bm\chi}(t)=\tau\dot{\bm\rho}(t)$. \,Then $\,{\bm p}(t)=({\bm\rho}(t),
{\bm\chi}(t))\,$
satisfies Eq.\,(\ref{contr}) with $\,u^n_m(t)\equiv0\,$ 
for $\,0\leq t\leq\epsilon$, \,with the correct initial condition
at $\,t=0$. Note that 
\qq
({\bm\rho}(\epsilon),{\bm\chi}(\epsilon))\ =\ \big((1-
\ee^{-\frac{\epsilon}{\tau}}){\bm\chi}_0,\,\ee^{-\frac{\epsilon}{\tau}}
{\bm\chi}_0\big)\,.
\label{initcond1}
\qqq
Since, by the assumptions, $\,{\bm\chi}_0\not=0$, \,we infer that
$\,{\bm\rho}(\epsilon)\not={0}\,$ and the solution of Eq.\,(\ref{contr})
for $\,\epsilon\leq t\leq T\,$ may be constructed as in the previous
point but taking (\ref{initcond1}) as the initial conditions at
$\,t=\epsilon$.  
\vskip 0.1cm

Similarly, if $\,{\bm\rho}_0\not={0}={\bm\rho}_1\,$ then 
set for $T-\epsilon\leq t\leq T$
\qq
{\bm\rho}(t)\ =\ \big(1-\ee^{\frac{T-t}{\tau}}\big){\bm\chi}_1\,,
\qqq
and $\,{\bm\chi}(t)=\tau\dot{\bm\rho}(t)$.
\,Then $\,{\bm p}(t)=({\bm\rho}(t),{\bm\chi}(t))\,$
satisfies Eq.\,(\ref{contr}) with $\,u^n_m(t)\equiv0\,$ 
for $\,T-\epsilon\leq t\leq T$, \,with the correct final condition
at $\,t=T$. \,One has 
\qq
({\bm\rho}(T-\epsilon),{\bm\chi}(T-\epsilon))\ =\ \big((1-
\ee^{\frac{\epsilon}{\tau}}){\bm\chi}_1,\,\ee^{\frac{\epsilon}{\tau}}
{\bm\chi}_1\big)\,.
\label{initcond2}
\qqq
Since, by the assumptions, $\,{\bm\chi}_1\not=0\,$ now, \,we infer that
$\,{\bm\rho}(T-\epsilon)\not={0}\,$ and the solution of Eq.\,(\ref{contr})
for $\,0\leq t\leq T-\epsilon\,$ with $\,{\bm\rho}(t)\not=0\,$ may be 
constructed as in the first point but taking (\ref{initcond2}) as the final 
condition at $\,t=T-\epsilon$.  
\vskip 0.1cm

Finally, if $\,{\bm\rho}_0={0}={\bm\rho}_1$, \,we combine the above
solutions for $\,0\leq t\leq\epsilon\,$ and $\,T-\epsilon\leq t\leq T\,$
with vanishing $\,u^n_m\,$ with the a solution with $\,{\bm\rho}(t)
\not={0}\,$ and appropriate $\,u^n_m(t)\,$ for 
$\,\epsilon\leq t\leq T-t$.
\end{proof}

\begin{rem}
\label{rem:nonzero}
Note that the solution $\,{\bm p}(t)\,$ of the ODE (\ref{contr})
satisfying $\bm{p}(0)=\bm{p}_{0}\neq 0$ and $\bm{p}(T)=\bm{p}_{1}\neq 0$ is everywhere nonzero.
\end{rem}

\nsection{{\bf PROJECTION OF THE DISPERSION TO} $\ S^{2d-1}$}
\label{sec:proj}

\noindent The generator $\,L\,$ of the process commutes with the 
multiplicative action of $\,\mathbb{R}_+\,$ on $\,\mathbb{R}^{2d}\,$ given 
by
\qq
{\bm p}\ \mathop{\longmapsto}\limits^{\Theta_\sigma}\ \sigma
\hspace{0.01cm}{\bm p}
\qqq
for $\,\sigma>0$. \,It follows that if $\,{\bm p}(0)\not=0\,$ 
then the projection 
\qq
[{\bm p}(t)]\,\equiv\,{\bm\pi}(t)
\qqq
of the process $\,{\bm p}(t)\,$ on the quotient space 
$\,\mathbb{R}^{2d}_{\not=0}/\mathbb{R}_+\,$ is also 
a Markov
process whose generator may be identified with $\,L\,$ acting on functions
on $\,\mathbb{R}^{2d}_{\not=0}\,$ that are homogeneous of degree zero.
The quotient space $\,\mathbb{R}^{2d}_{\not=0}/\mathbb{R}_+\,$
may be naturally identified with the sphere
\qq
S^{2d-1}\ =\ \{\,({\bm\rho},{\bm\chi})\ |\ {\bm\rho}^2+{\bm\chi}^2=R^2\,\}
\label{sph}
\qqq
for a fixed $\,R\,$ and we shall often use this identification below.
\,The transition probabilities $\,P_t({\bm\pi}_0;d{\bm\pi})\,$ of the 
process $\,{\bm\pi}(t)\,$ are obtained by projecting the original transition 
probabilities from $\,\mathbb{R}^{2d}_{\not=0}\,$ to 
the quotient space.
\,Note that the vector fields $\,X_0,\ X^m_n,\ Y^m_n,\ Z^m_n\,$ also commute 
with the action $\,\mathbb{R}_+\,$ so may be identified with vector fields 
on $\,\mathbb{R}^{2d}_{\not=0}/\mathbb{R}_+\,$ 
and Eq.\,(\ref{LX}) still holds. Viewed as vector fields
on $\,S^{2d-1}$, $\,X^m_n,\ Y^m_n\,$ and $\,Z^m_n\,$ still span at each point
the tangent space to $\,S^{2d-1}$. \,It follows that the operators 
$\,L,\ L^\dagger,\ \partial_t-L,\ \partial_t-L^\dagger,\ 2\partial_t-L\otimes 1
-1\otimes L^\dagger\,$ (with the adjoints defined now with respect to 
an arbitrary measure with smooth positive density on 
$\,S^{2d-1}$, \,e.g. the normalized standard $\,SO(2d)$-invariant one
$\,\mu_0(d{\bm\pi})$) are still hypoelliptic and the transition probabilities 
of the projected process have smooth densities 
$\,P_t({\bm\pi}_0;{\bm\pi})\,$ with respect to $\,\mu_0(d{\bm\pi})\,$
for $\,t>0$. \,Consequently, the process $\,{\bm\pi}(t)\,$ 
is {\,\bf strongly Feller}: \,for bounded measurable functions $\,f\,$
on $\,S^{2d-1}$, \,the functions
\qq (T_tf)({\bm\pi}_0)\ =\ \int\limits_{S^{2d-1}}
P_t({\bm\pi}_0;d{\bm\pi})\,f({\bm\pi})\ =\  
\int\limits_{S^{2d-1}}
P_t({\bm\pi}_0;{\bm\pi})\,f({\bm\pi})\,\mu_0(d{\bm\pi})
\label{Tt}
\qqq
are continuous (and even smooth) for $\,t>0$.
\,Besides, the projected process is still irreducible since
$\,P_t({\bm\pi}_0;{\bm\pi})>0\,$ for all $\,t>0\,$ and $\,{\bm\pi}_0,{\bm\pi}
\in S^{2d-1}$. \,The latter property follows from the relation between 
$\,P_t({\bm\pi}_0;{\bm\pi})\,$ and $\,P_t({\bm p}_0;{\bm p})\,$ and 
from the strict positivity of the latter away from the origin of 
$\,\mathbb{R}^{2d}$.
\vskip 0.1cm

The gain from projecting the process $\,{\bm p}(t)\,$
to the compact space $\,S^{2d-1}\,$ is that the projected process 
$\,{\bm\pi}(t)\,$ 
has necessarily invariant probability measures $\,\mu(d{\bm\pi})$. 
\,In particular, each weak-topology accumulation point 
for $\,T\to\infty\,$ of the Cesaro means
\qq
T^{-1}\int\limits_0^T P_t({\bm\pi}_0;d{\bm\pi})\,dt
\qqq
provides such a measure\footnote{Probability measures on a compact space
form a compact set in weak topology.}. \,Since the 
({\it a priori} distributional) density
$\,n({\bm\pi})\,$ of an invariant measure is annihilated by 
$\,L^\dagger$, \,the hypoellipticity of the latter operator assures that 
$\,n({\bm\pi})\,$ is a smooth function. The invariance relation
\qq
\int\limits_{S^{2d-1}}P_t({\bm\pi}_0,{\bm\pi})\,n({\bm\pi}_0)\,
\mu_0(d{\bm\pi}_0)\ =\ n({\bm\pi})
\qqq
together with the strict positivity of $\,P_t({\bm\pi}_0,{\bm\pi})\,$
implies then the strict positivity of the density $\,n({\bm\pi})\,$ 
of the invariant measure and, in turn, the uniqueness of the latter 
(different ergodic invariant measures have to have disjoint
supports, so that there may be only one such measure), see e.g.
\cite{R-B} for more details. \,One obtains this way
\vskip 0.3cm

\begin{theo}\ \ The projected process 
$\,{\bm\pi}(t)\,$ has a unique invariant probability measure 
$\,\mu(d{\bm\pi})\,$ with a smooth strictly positive density $\,n({\bm\pi})$.
\end{theo}

The smoothness of the densities $\,P_t({\bm\pi}_0;{\bm\pi})\,$ implies
by the Arzel\`{a}-Ascoli Theorem that the operators of the semigroup $\,T_t\,$ on
the space $\,C(S^{2d-1})\,$ of continuous function on $\,S^{2d-1}\,$ with the
$\,\sup$-norm, \,defined by Eq.\,\eqref{Tt}, \,are compact for $\,t>0$. 
\,The uniqueness of the invariant measure implies 
then that the spectrum of $\,T_t\,$ is strictly inside the unit disk 
except for the geometrically simple eigenvalue $\,1\,$ corresponding to
the constant eigenfunctions, see \cite{R-B}. It follows 
that the process $\,{\bm\pi}(t)\,$ 
is exponentially mixing:

\begin{theo}
\qq
\Big\langle f_1({\bm\pi}(t_1))\,f_2({\bm\pi}(t_2))\Big\rangle\ \ 
\mathop{\longrightarrow}\limits_{t_1\to\infty\atop t_2-t_1\to\infty}\ \ 
\int f_1({\bm\pi})\,\mu(d{\bm\pi})\,\int f_2({\bm\pi})\,\mu(d{\bm\pi})
\qqq
exponentially fast for continuous functions $\,f_1,f_2$.
\end{theo}

\nsection{\bf{PROPERTIES OF THE INVARIANT MEASURE}}

\noindent Due to the isotropy of the covariance (\ref{icov}),
\,the generator $\,L\,$ of the process $\,{\bm\pi}(t)\,$ commutes with the 
action of the rotation group \,$SO(d)\,$ induced 
on $\,S^{2d-1}\,$ by the mappings
\qq
({\bm\rho},{\bm\chi})\ \ \mathop{\longmapsto}\limits^{\Theta_O}\ \ 
(O{\bm\rho},O{\bm\chi})
\qqq
for $\,O\in SO(d)$.  \,As a consequence, the process $\,{\bm\pi}(t)\,$ stays 
Markov when projected to the quotient
space $\,P_d=S^{2d-1}/SO(d)$. \,The unique invariant measure 
$\,\mu(d{\bm\pi})\,$ of the process $\,{\bm\pi}(t)\,$ has to be 
also invariant under $\,SO(d)\,$ and its projection to $\,P_d\,$
provides the unique invariant probability measure of the projected
process\footnote{To see the uniqueness, note that averaging over 
the action of $\,SO(d)\,$ maps $\,C(S^{2d-1})\,$ to $\,C(P_d)\,$ 
and that dual map sends
invariant measures for the projected process to invariant 
measures of $\,{\bm\pi}(t)$.}.    
\,The projected invariant measure may be expressed in terms 
of invariants of the $\,SO(d)$-action. \,Such invariants will be chosen 
as the following dimensionless combinations:
\begin{itemize}
\item for $\,d=1\,$ where $\,P_1=S^1$
\qq
x\,=\,\frac{\chi}{\rho}\,,
\qqq
\item
for $\,d=2\,$ where $\,P_2=\mathbb{P}\mathbb{C}^1$, 
\qq
\label{inv2}
x\,=\,\frac{{\bm\rho}\cdot{\bm\chi}}{{\bm\rho}^2}\qquad{\rm and}\qquad 
y\,=\,\frac{\rho^1\chi^2-\rho^2\chi^1}{{\bm\rho}^2}
\qqq
with $\,z=x+iy\,$ providing the inhomogeneous complex coordinate of
$\,\mathbb{P}\mathbb{C}^1$,
\item
for $d\geq 3$,
\qq
\label{inv3}
x\,=\,\frac{{\bm\rho}\cdot{\bm\chi}}{{\bm\rho}^2}\qquad{\rm and}\qquad 
y\,=\,\frac{\sqrt{{\bm\rho}^2{\bm\chi}^2-({\bm\rho}\cdot{\bm\chi})^2}}
{{\bm\rho}^2}\,.
\label{invs3}
\qqq
\end{itemize}
Note that the right hand side of the $\,d\geq3\,$ expression for $\,y\,$
would give in $\,d=2\,$ the absolute value of $\,y$. 
\,The quotient spaces
$\,P_d\,$ are not smooth for $\,d\geq3$.

\subsection{$d=1\,$ case}
\label{subsec:ergod1}

\noindent In one dimension, \,Eq.\,(\ref{inteq2}) implies that
\qq
dx\ =\ -\frac{_1}{^\tau}\big(x+x^2\big)dt\,+\,dS(t)\,.
\label{SDE1}
\qqq
The invariant probability measure on $\,S^1\,$ is easily found 
\cite{Halperin,WM1} to have the form $\,d\mu=\eta(x)\hspace{0.02cm}dx\,$ 
with
\qq
\eta(x)\ =\ Z^{-1}\Big(\ee^{-\frac{1}{\tau(A+2B)}\big(\frac{2}{3}x^3
+x^2
\big)}\int\limits_{-\infty}^x\ee^{\frac{1}{\tau(A+2B)}\big(\frac{2}{3}x'^3+
x'^2\big)}dx'\Big)\hspace{0.03cm}dx\,,
\label{mu1}
\qqq
where $\,Z\,$ is the normalization constant. Since the normalized 
rotationally invariant
measure on $\,S^1=\{(\rho,\chi)\ |\ \rho^2+\chi^2=R^2\}\,$ has the form
$\,d\mu_0=\frac{dx}{\pi(1+x^2)}$, \,it follows from our general
result that the density $\,n(x)=\pi(1+x^2)\hspace{0.03cm}\eta(x)\,$ 
of the invariant measure
relative to $\,d\mu_0\,$ must be smooth and positive
at $\,x=\infty$, \,i.e. at the origin when expressed in the variable 
$\,x^{-1}$. \,In particular,
\qq
\eta(x)\ =\ \CO(|x|^{-2})\qquad{\rm for}\qquad |x|\to\infty\,,
\qqq
which may also be easily checked directly.   
\vskip 0.1cm

In one dimension, the generator $\,L\,$ given by 
Eq.\,(\ref{genL}) acts on a function $\,f(x)\,$ according to the formula:
\qq
(Lf)(x)\ =\ -\frac{_1}{^\tau}(x^2+x)\,\partial_xf(x)\,+\,\frac{_1}{^2}(A+2B)
\,\partial_x^2f(x)\,.
\label{Lonf1}
\qqq
It coincides with the generator of the process satisfying 
the SDE (\ref{SDE1}). \,The trajectories of the latter process 
with probability one explode to $\,-\infty\,$ in finite time but,
in the version of the process that describes the projectivized
dispersion of the one-dimensional inertial particle, they re-enter 
immediately from $\,+\infty$.

\subsection{$d=2\,$ case}
\label{subsec:ergod2}

\ 

\vskip 0.1cm

\noindent In two dimensions, \,the invariant measure on $\,S^3\,$ has 
to have the form
\qq
d\mu\ =\ \frac{_1}{^{2\pi}}\,\eta(z,\bar z)\,
d^2z\,d\hspace{0.04cm}{\rm arg}({\bm\rho})\,.
\label{2dinvm}
\qqq
On the other hand, the $\,SO(4)$-invariant normalized measure on $\,S^3\,$ is
\qq
d\mu_0\ =\ \frac{_1}{^{2\pi}}\,\eta_0(z,\bar z)\,d^2z\,d\hspace{0.04cm}
{\rm arg}({\bm\rho})
\qqq
with
\qq
\eta_0(z,\bar z)\ =\ \frac{1}{\pi(1+|z|^2)^2}\,.
\qqq
It follows from the general result obtained above that the density 
of $\,d\mu\,$ relative to $\,d\mu_0\,$
\qq
n(z,\bar z)\ =\ \frac{\eta(z,\bar z)}{\eta_0(z,\bar z)}
\qqq
has to extend to a smooth positive function on 
$\,\mathbb{P}\mathbb{C}^1$, \,i.e. 
to be smooth and positive at zero when expressed in the variables 
$\,(z^{-1},\bar z^{-1})$. 
\,In particular, 
\qq
\eta(z,\bar z)\ =\ \CO(|z|^{-4})\qquad{\rm for}\qquad|z|\to\infty\,.
\label{2ddens}
\qqq
The unique invariant probability measure of the Markov process obtained 
by projecting $\,{\bm\pi}(t)\,$ from $\,S^3\,$ to $\,
S^3/SO(2)=\mathbb{P}\mathbb{C}^1\,$ has the form (\ref{2dinvm}) with 
$\,\frac{1}{2\pi}d\hspace{0.04cm}{\rm arg}({\bm\rho})\,$ on the right 
hand side dropped. \,Note that the relation (\ref{2ddens}) implies that
\qq
\int\limits_{-\infty}^{\infty}\eta(x,y)\,dy\ =\ \CO(|x|^{-3})
\qquad{\rm for}\qquad|x|\to\infty\,.
\label{bound2}
\qqq
by changing variables $\ y\mapsto\sqrt{1+x^2}\,y\ $ in the integral.
Such behavior was heuristically argued for and numerically checked 
in \cite{BCH1}. 
\vskip 0.1cm

In two dimensions, the generator $\,L\,$ of 
Eq.\,(\ref{genL}) acts on $\,SO(d)\,$ invariant functions $\,f(x,y)\,$ 
according to the formula:
\qq
(Lf)(x,y)&=&-\,\frac{_1}{^\tau}(x^2-y^2+x)\,\partial_xf(x,y)\cr
&&-\,\frac{_1}{^\tau}(2xy+y)\,
\partial_yf(x,y)\cr
&&+\,\frac{_1}{^2}(A+2B)\,\partial_x^2f(x,y)\,+\,\frac{_1}{^2}A
\,\partial_y^2f(x,y)\,.
\label{Lonf2}
\qqq
It coincides with with the generator of the process $\,z(t)=(x+iy)(t)\,$ 
in the complex plane given by the SDE \cite{Piterbarg}
\qq
dz\ =\ -\frac{_1}{^\tau}\big(z+z^2\big)\hspace{0.02cm}dt\,+\,
\sqrt{A+2B}\,d\beta^1(t)\,+\,i\sqrt{A}\,d\beta^2(t)\,,
\label{2dequ}
\qqq
where $\,\beta^1(t)\,$ and $\,\beta^2(t)\,$ are two independent Brownian
motions.

\subsection{$d\geq3\,$ case}
\label{subsec:ergod3}

\ 

\vskip 0.1cm

\noindent Finally, in three or more dimensions, \,the invariant measure 
on $\,S^{2d-1}\,$ has to have the form
\qq
d\mu\ =\ \eta(x,y)\,dx\hspace{0.03cm}dy\,d[O]\,,
\label{3dinvm}
\qqq
where $\,O\in SO(d)\,$ is the rotation matrix such that $\,O^{-1}{\bm\rho}\,$
is along the first positive half-axis in $\,\mathbb{R}^d\,$
and $\,O^{-1}{\bm\chi}\,$ lies in the 
half-plane spanned by the first axis and the second positive half-axis.
Note that, generically, $\,O\,$ is determined modulo rotations in $\,(d-2)\,$
remaining directions. $\,d[O]\,$ stands for the normalized 
$\,SO(d)$-invariant measure on $\,SO(d)/SO(d-2)$. \,In the same notation,
the $\,SO(2d)$-invariant normalized measure on 
$\,S^{2d-1}\,$ takes the form
\qq
d\mu_0\ =\ \eta_0(x,y)\,dx\hspace{0.03cm}dy\,d[O]\,.
\label{dmu0}
\qqq
for
\qq
\eta_0(x,y)\ =\ \frac{(d-1)2^{d-1}y^{d-2}}{\pi(1+x^2+y^2)^d}\,,
\label{eta0}
\qqq
as is shown in Appendix \ref{app:A}.
\,As before, it follows from the general analysis that the function
\qq
n(x,y)\ =\ \frac{\eta(x,y)}{\eta_0(x,y)}
\qqq
is smooth and positive on the sphere $\,S^{2d-1}=\{({\bm\rho},{\bm\chi})\ |\ 
{\bm\rho}^2+{\bm\chi}^2=R^2\}$. \,In particular, this implies that
\qq
\eta(x,y)\ =\ \CO(y^{d-2})\qquad{\rm for}\qquad y\searrow 0
\qqq
i.e. for $\,{\bm\rho}\,$ and $\,{\bm\chi}\,$ becoming parallel or
$\,{\bm\chi}^2\,$ becoming small  and 
\qq
\eta(x,y)\ =\ \CO(|x|^{-2d})\qquad{\rm for}\qquad |x|\to\infty 
\qqq
when $\,{\bm\rho}^2\to0\,$ but the angle between $\,{\bm\rho}\,$ and 
$\,{\bm\chi}\,$ stays away from a multiple of $\,\frac{\pi}{2}$.
\,The smoothness and positivity of $\,n(x,y)\,$ on $\,S^{2d-1}\,$
imply \,(again by changing variables $\ y\mapsto\sqrt{1+x^2}\,y\ $ 
in the integral) \,that now
\qq
\int\limits_0^\infty\eta(x,y)\,dy\ =\ \CO(|x|^{-d-1})\qquad{\rm for}\qquad 
|x|\to\infty\,.
\label{bound3}
\qqq
\vskip 0.2cm

A straightforward calculation shows that, in three or more dimensions,
\,the action on $\,L\,$ on $\,SO(d)$-invariant functions $\,f(x,y)\,$ 
is given by a generalization of Eq.\,(\ref{Lonf2}):
\qq
(Lf)(x,y)&=&-\,\frac{_1}{^\tau}(x^2-y^2+x)\,\partial_xf(x,y)\cr
&&-\,\frac{_1}{^\tau}(2xy+y-\frac{_{\tau A(d-2)}}{^{2y}})\,
\partial_yf(x,y)\cr
&&+\,\frac{_1}{^2}(A+2B)\,\partial_x^2f(x,y)\,+\,\frac{_1}{^2}A
\,\partial_y^2f(x,y)\,.
\label{Lonf3}
\qqq
It coincides with the generator of the process
$\,z(t)=(x+iy)(t)\,$ in the complex plane given by the SDE \cite{BCH1,Horvai}
\qq
dz\ =\ -\frac{_1}{^\tau}\big(z+z^2-i\frac{_{\tau A(d-2)}}{^{2\,{\rm Im}(z)}}
\big)\hspace{0.02cm}dt\,+\,\sqrt{A+2B}\,d\beta^1(t)\,+\,i\sqrt{A}\,d\beta^2(t)
\label{3dequ}
\qqq
which upon setting $\,d=2\,$ reduces to the SDE (\ref{2dequ}).

\nsection{{\bf ABSENCE OF EXPLOSION IN THE COMPLEX (HALF-)PLANE}}
\label{sec:d>=2}

%This also will have to be the property of the process $\,z(t)\,$ obtained  by projecting $\,{\bm\pi}(t)\in S^3\,$  to $\,S^3/SO(2)=\mathbb{P}\mathbb{C}^1\,$ and expressing the projected process in the inhomogeneous coordinate $\,z=x+iy$.  \,Indeed, \,the latter process has the same generator, hence the same law. Recall that the projected process has the unique invariant probability measure $\,\eta(z,\bar z)\,d^2z\,$ with $\,\eta(z,\bar z)>0\,$ discussed in Sec.\,\ref{subsec:ergod2}. It follows that this measure is also the unique invariant probability measure of the process in the complex plane defined by the SDE (\ref{2dequ}). The non-explosiveness of the $\,d=2\,$ process $\,z(t)\,$ should be contrasted with the explosive behavior of its one-dimensional counterpart $\,x(t)$, \,see Sec.\,\ref{subsec:ergod1}.   

\noindent Let us set 
\begin{equation}
Q_d=\begin{cases}
\,\mathbb{R}^2 & \text{ if }\quad d =2\,,\\
\,\mathbb{H}_{+} & \text{ if }\quad d\geq 3\,,
\end{cases}
\end{equation}
where 
\begin{equation*}
\mathbb{H}_{+}=\{(x,y)\ |\ y>0 \}
\end{equation*}
is the upper-half plane. \,Note that $\,Q_d\,$ may be 
identified with an open dense subset of the quotient space 
$\,P_d=S^{2d-1}/SO(d)\,$ using the $\,SO(d)$-invariants 
(\ref{inv2}) or (\ref{inv3}) on $\,S^{2d-1}$. \,We shall often use 
the complex combination $\,x+iy\,$ as a coordinate on $\,Q_d$.
\vskip 0.1cm

In the present section,
we shall show that for $\,d\geq 2\,$ the unique solution 
of the SDE (\ref{3dequ}) starting from $\,z\in Q_d\,$ remains in 
$\,Q_d\,$ for all times $t\geq 0$ with probability one. \,This will
also have to be the property of the projection of the process $\,{\bm\pi}(t)\,$
to the quotient space $\,P_d=S^{2d-1}/SO(d)\,$ when described 
in the complex coordinate $\,z=x+iy$. \,Indeed, the coincidence of
the generators of the two processes will assure that they have the same
law. \,Let us start by generalizing and simplifying \eqref{3dequ}. 
\vskip 0.1cm

Let $\,w(t)=z(t)+1/2$, where $\,z(t)\,$ solves 
\eqref{3dequ} with $z(0)=x+iy\in Q_d$.  \,Clearly, $\,w(t)\,$ satisfies 
an SDE of the form 
\begin{equation}
\label{sdew}dw=\frac{_1}{^\tau}\left(-w^{2}+\alpha +i\frac{_{\tau b(d-2)}}
{^{\rm{Im}(w)}}\right)dt + \sqrt{\frac{_{2\kappa_{1}}}{^\tau}} 
d\beta^{1}(t) + i\sqrt{\frac{_{2\kappa_{2}}}{^\tau}}d\beta^{2}(t),
\end{equation}
where $\alpha=a_{1}+ia_{2}\in \mathbb{C}$, $b>0$, $\kappa_{1}\geq 0$, $\kappa_{2}>0$, and $\beta^{1}(t)$ and $\beta^{2}(t)$ are two independent Brownian 
motions.  When $\,d=2$, \,the term proportional to $\,b(d-2)\,$ is
absent from \eqref{sdew}. \,When $\,d\geq 3$, \,we suppose that 
$\,\tau b(d-2)\geq \kappa_{2}$.  \,Clearly in \eqref{3dequ}, all of these 
assumptions are met under the given substitution.  Since $\,w(t)\,$ is a 
horizontal shift of $\,z(t)$, $\,w(t)\,$ stays in $\,Q_d\,$ 
with probability one for all times if and only if $\,z(t)\,$ does. 
\,Employing methods of refs. \cite{R-B, Hasminskii-1980, MT3}, we shall 
estimate the time at which the process $\,w(t)\,$ leaves $\,Q_d$.  
\,To this end, it is easy to see that there exists a sequence of precompact 
open subsets $\{ O_{n}\,|\,n\in \mathbb{N}\}$ of $\,Q_d\,$ such that 
$\,O_{n}\uparrow Q_d\,$ as $\,n\rightarrow \infty$.  \,Thus we may 
define stopping times:
\begin{equation}
\tau_{n}= \inf\{s>0\ |\ w(s)\in O^{c}_{n} \},
\end{equation}          
for $\,n\in \mathbb{N}$.  \,Let $\,\tau_\infty\,$ be the finite 
or infinite limit 
of $\,\tau_{n}\,$ as $\,n\rightarrow \infty$.  

\begin{defn}
We say that the solution $\,w(t)\,$ is \emph{non-explosive} if 
\begin{equation}\label{EXP}
P\left[\tau_\infty=\infty  \right]=1.
\end{equation}
\end{defn}
\noindent Naturally, in order to show that $\,w(t)\,$ remains in 
$\,Q_d\,$ for $\,t\geq 0\,$ with probability one, it is enough 
to prove that $\,w(t)\,$ is non-explosive.  

Let $\,M\,$ be the generator of the process $\,w(t)=x(t)+iy_{t}$.  \,We see 
that for $\,f\in C^{\infty}(Q_{d})\,$:
\begin{eqnarray}
\label{genw}(M_{d}f)(x,y)&=&-\frac{_1}{^\tau}(x^{2}-y^{2}-a_{1}) 
\partial_{x}f(x,y) -\frac{_1}{^\tau}(2xy-a_{2}-  \tau b (d-2)y^{-1}) 
\partial_{y}f(x,y)\\
\nonumber &\,&+\, \frac{_{\kappa_{1}}}{^\tau} \partial_{x}^{2}f(x,y) + 
\frac{_{\kappa_{2}}}{^\tau} \partial_{y}^{2}f(x,y),
\end{eqnarray}
where the term $\tau b(d-2) y^{-1}$ is absent for $d=2$.  
Let us define
\begin{equation}
\partial Q_d=
\begin{cases}
\infty &\text{ if\quad} d=2\,,\\
\{(x,y) \in \mathbb{R}^2\ |\ y=0 \}\cup  \infty &\text{ if\quad} d\geq 3\,,
\end{cases}
\end{equation}
with $\,\infty\,$ denoting the point compactifying $\,\mathbb{R}^2$.
\,To ensure condition \eqref{EXP}, it suffices to construct a (Lyapunov) function 
$\,\Phi_{d}\in \mathbb{C}^{\infty}(Q_d)$ that satisfies:
\begin{enumerate}[(I)]
\item \label{I}$\Phi_{d}(x,y) \geq 0\,$ for all $\,(x,y)\in Q_d\,$,
\item \label{II} $\Phi_{d}(x,y) \rightarrow \infty\,$ as $\,(x,y) \rightarrow 
\partial Q_d$, $\,(x,y)\in Q_d\,$,
\item \label{III} $M_{d} \Phi_{d}(x,y) \leq C \Phi_{d}(x,y)\,$ 
for all $\,(x,y)\in Q_{d}$, \,where $\,C>0\,$ is a positive constant.  
\end{enumerate}
See, for example, \cite{MT3}.  We will show:
\begin{theo}\label{TH1}
If $\,\kappa_1\geq0\,$ and $\,\tau b(d-2)\geq\kappa_2>0\,$ then
there exists $\,\Phi_{d}\in C^{\infty}(Q_d)\,$ that satisfies 
\eqref{I}, \eqref{II}, and 
\begin{equation}
\label{IV}\tag{IV} M_{d} \Phi_{d}(x,y) \rightarrow -\infty \text{\ as\ } 
(x,y)\rightarrow \partial Q_d,\ (x,y)\in Q_d.
\end{equation}
\end{theo} 

\noindent Given such $\,\Phi_{d}$, \,clearly $\,\Phi_d+1\,$ will satisfy 
\eqref{I}, \eqref{II} and \eqref{III}.  We will then have: 

\begin{theo}\label{TH2}
Under the assumptions of Theorem \ref{TH1}, \,the solution $\,w(t)\,$ of 
the SDE \eqref{sdew} stays in $\,Q_d\,$ for all
times $\,t> 0\,$ with probability one if $\,w(0)=x+iy\in Q_d$.
%
%For all $(x,y)\in R(d)$ with $z_{0}(d)=x+iy$, $z_{t}(d) \in R(d)$ for 
%all times $t\geq 0$ with probability one. 
\end{theo}

\begin{cor}
\label{cor:nonexpl}
This implies the same result about the solution $\,z(t)\,$ of
the SDE \eqref{3dequ} with $\,A>0\,$ and $\,A+2B\geq0$.
\end{cor}

The existence of the Lyapunov function with the properties 
asserted in Theorem \ref{TH1} has another consequence. It
allows to show that
\qq
\lim\limits_{n\to\infty}\,\liminf\limits_{T\to\infty}\,\frac{_1}{^T}\int_0^T
P_t(w,O_n^c)\,dt\,=\,0
\qqq
for the SDE \eqref{sdew} and $\,O_{n}\uparrow Q_d\,$ as before,
implying the existence of an invariant measure on $\,Q_d$,
\,see Theorems 4.1 and 5.1 in Chapter III of \cite{Hasminskii-1980}. 
If the generator of the process is
elliptic, then the same tools that we used for the projectivized dispersion
(i.e. hypoellipticity and control theory \cite{R-B}) show that the 
invariant measure must have a smooth strictly positive 
density and be unique. \,This gives:

\begin{theo}\label{TH3} 
Under the assumptions of Theorem \ref{TH1}, \,the system \eqref{sdew} on
$\,Q_d\,$ has an invariant measure which is unique and has a smooth strictly 
positive density if $\,\kappa_{1}>0$. 
\end{theo}

\begin{rem}
\label{rem:uniq}
Theorem \ref{TH3} allows to reaffirm and strengthen
what has already been proven earlier since it implies 
the existence of an invariant measure for the system \eqref{3dequ} 
if $\,A>0\,$ and $\,A+2B\geq0\,$ and its uniqueness if $\,A+2B>0$. 
\,Given the non-explosivity result of Corollary \ref{cor:nonexpl}, 
\,the approach taken earlier implied the existence 
and the uniqueness of an invariant measure for the system \eqref{3dequ} 
under more stringent conditions: $\,A>0$, $\,A\geq|B|\,$ and 
$\,A+(d+1)B\geq0$.
\end{rem}
  
The construction of the Lyapunov function $\,\Phi_{d}\,$ is split up into 
two cases: $\,d=2\,$ and $\,d\geq 3$.  
\,The existence of $\Phi_{d}$ for $d\geq 3$ will be easy, given
$\,\Phi_{2}$.  Thus we shall first construct $\,\Phi_{2}$.  

\subsection{$d=2\,$ case}  
\label{subsec:d=2}

\ 

\vskip 0.1cm

\noindent It is not easy to write down a globally defined function 
$\,\Phi_{2}\,$ that satisfies \eqref{I}, \eqref{II}, and \eqref{IV} 
in all of $\,Q_2=\mathbb{R}^{2}$.  \,This is because the signs of the 
coefficients of the vector fields in $\,M_{2}\,$ vary over different regions 
in $\,\mathbb{R}^{2}$.  \,We shall thus construct functions that satisfy these 
properties in different regions, the union of which is $\,\mathbb{R}^{2}$.  
\,We shall then glue together these functions to form one single globally 
defined $\,\Phi_{2}$.  \,One should note that this idea is similar in spirit 
to that of M. Scheutzow in \cite{MS}.  \,Let $\,r=\sqrt{x^{2}+y^{2}}$.  \,For 
the rest of Subsection \ref{subsec:d=2}, we will drop the use of the 
subscript $2$ in $\,M_{2}\,$ and $\,\Phi_{2}$.  \,We first need the following:     

\begin{defn}
Let $\,X\subset \mathbb{R}^{2}\,$ be unbounded.  We say that a function 
$\,f(x,y) \rightarrow \pm \infty\,$ \emph{as} $\,r\rightarrow \infty\,$ 
\emph{in} $\,X\,$ if $\,f(x,y)\rightarrow \pm \infty\,$ as $\,(x,y) \to 
\infty$, $\,(x,y) \in X$.
\end{defn}      

\begin{defn}\label{DEF}
Let $\,\,X\subset \mathbb{R}^{2}\,$ be unbounded and let 
$\,\varphi\in C^{\infty}(X)\,$ satisfy 
\begin{enumerate}[(i)]
\item $\varphi \geq 0\,$ for all $\,(x,y)\in X$,
\item $\varphi \rightarrow \infty\,$ as $\,r\rightarrow \infty\,$ in $\,X$,
\item $M \varphi \rightarrow -\infty\,$ as $\,r\rightarrow \infty\,$ 
in $\,X$.
\end{enumerate}
We call $\varphi$ a \emph{Lyapunov function in} $\,X\,$ \emph{corresponding 
to} $\,M\,$ and denote
\begin{equation*}
\mathcal{N}(\alpha,\kappa_{1},\kappa_{2}, X)\ =\ 
\left\{\text{\,Lyapunov functions 
in } X \text{\ corresponding to } M\, \right\}.
\end{equation*}
We shall abbreviate ``Lyapunov function'' by LF.  
\end{defn}

\begin{defn}
Let $\,X\subset \mathbb{R}^{2}\,$ be unbounded and 
$\,f,g:X\rightarrow \mathbb{R}$.  \,We shall say that $\,f\,$ 
\emph{is asymptotically equivalent to} $\,g\,$ \emph{in} $\,X\,$ 
and write $\,f\simeq_{X} g\,$ if $$\lim_{r\rightarrow \infty}\,\frac{f(x,y)}
{g(x,y)}\,=\,1\,,$$ where the limit is taken only over points $\,(x,y)\in X$. 
\end{defn}

%=======================================================================
%----------------------------------------------Lyapunov Functions d=2--------------------------------------------------
%======================================================================

It is clearly sufficient to construct LFs in regions that 
cover $\,\mathbb{R}^{2}$, \,except, possibly, a large ball about the origin.  
The constructions will be done in a series of propositions.  The possibly 
daunting multitude of parameters is designed to make the gluing possible.  
There is a total of five LFs in five different regions and 
the details that follow are not difficult to verify.  The crucial LF is 
the fifth one, $\,\varphi_{5}$, \,defined in a region where 
explosion occurs in a nonrandom equation, i.e., when 
$\,\alpha=\kappa_{1}=\kappa_{2}=0\,$ in \eqref{sdew}.  

\begin{prop}
Let $\,X_{1}=\{x\geq 1\}\subset \mathbb{R}^{2}$, $\,C_{1}>0$, \,and 
$\,\delta \in(0, 1/2)$.  \,Define 
\begin{equation}
\varphi_{1}(x,y)=C_{1}(x^{2}+y^{2})^{\delta/4}\tag{LF1}.
\end{equation}  
We claim that $\,\varphi_{1} \in \mathcal{N}(\alpha,\kappa_{1}, \kappa_{2},
X_{1})\,$ for all $\,\alpha \in \mathbb{C}$, $\,\kappa_{1}\geq 0,\  
\kappa_{2} >0$.
\end{prop}
\begin{proof}  $\,\varphi_{1}$ is nonnegative everywhere in $\,\mathbb{R}^{2}$,
\,hence everywhere in $\,X_{1}$.  $\,\varphi \rightarrow \infty\,$ 
as $\,r \rightarrow \infty\,$ in all of $\,\mathbb{R}^{2}$, \,hence in all 
of $\,X_{1}$.  \,It is easy to check that $\,\partial_{xx}\varphi_{1}\,$ 
and $\,\partial_{yy}\varphi_{1}\,$ both go to zero as $\,r \rightarrow \infty$.
\,Thus dropping second order terms in the expression for $\,M\varphi_{1}$, 
\,we have 
\begin{eqnarray}
\nonumber \tau M\varphi_{1}&\simeq_{X_{1}}&-\,\frac{{C_{1}\delta}}{2}\,x(x^{2}
+y^{2})^{\delta/4}\,+\,\frac{{C_{1}\delta}}{2}\,\frac{a_{1}x+a_{2}y}
{(x^{2}+y^{2})^{1-\delta/4}}\\
\label{varphi1}&\simeq_{X_{1}}&-\,\frac{C_{1}\delta}{2}\,
x(x^{2}+y^{2})^{\delta/4}\ \rightarrow\ -\infty 
\end{eqnarray}as $\,r \rightarrow \infty\,$ in $\,X_{1}$, \,since 
$\,x\geq 1$ in $\,X_{1}$.  
\end{proof}
We need a remark before we move onto the next region.  Let $\,\mathbb{R}
\subset \mathbb{R}^{2}\,$ be the real axis.    

\begin{rem}\label{REM}
Let $\,f(x,y)=u(x, |y|)\,$ be a twice differentiable function in 
$\,X\setminus \mathbb{R}$.  \,Then 
\begin{eqnarray*}
(\tau Mf)(x,y)&=& \kappa_{1}\, u_{xx}(x, |y|)\,+\,\kappa_{2}\, 
u_{|y||y|}(x, |y|) \\
&\,&+\, (y^{2}-x^{2}+a_{1})\,u_{x}(x, |y|)\,+\, 
(-2x|y|+\text{sgn}(y)a_{2})\,u_{|y|}(x, |y|).
\end{eqnarray*}
\begin{proof}
Apply the chain rule to the operator $\,\partial_{|y|}$.  
\end{proof}
\end{rem}

\noindent In the following arguments, often the function will be of the form 
$\,f(x,y)=u(x, |y|)$.  \,The above remark will allow for simplifications 
in the argument for property (iii) in Definition \ref{DEF}.  

\begin{prop}
Let $C_{2} > 0$, $\,\delta \in (0, 1/2)\,$ and 
\begin{equation}
\varphi_{2}=C_{2}(-x + |y|^{\delta/2})\tag{LF2}.
\end{equation}  Then $\,\varphi_{2} \in \mathcal{N}(\alpha,\kappa_{1}, 
\kappa_{2}, X_{2})\,$ for all $\,\alpha\in \mathbb{C}\,$ and all 
$\,\kappa_{1}\geq 0$, $\kappa_{2} >0$, \,where 
$$X_{2}\,=\,\{-2 \leq x \leq 2\}\,\cap\,\{|y|\geq 2^{2/\delta}\}.$$
\end{prop}
\begin{proof}
$\,\varphi_{2}\,$ is indeed smooth in $\,X_{2}\,$ since $\,X_{2}\,$ 
is bounded away from $\,\mathbb{R}$.  \,Note that the region was chosen 
so that $\,\varphi_{2} \geq 0\,$ in $\,X_{2}$.  \,Moreover, since $\,x\,$ 
is bounded in this region, $\,r\rightarrow \infty$ in $\,X_{2}\,$ if and
only if
$\,|y|\rightarrow \infty$.  \,Hence, $\,\varphi_{2}\rightarrow \infty\,$ 
in $\,X_{2}$.  \,By Remark \ref{REM} and noting that $\,\partial_{xx} 
\varphi_{2}=0\,$ and that $\,\partial_{|y||y|} \varphi_{2}\rightarrow 0\,$ 
as $\,|y|\rightarrow \infty$, \,we have 
\begin{eqnarray*}
\tau M \varphi_{2}(x,y)&\simeq_{X_{2}}& C_{2}\,(x^{2}-y^{2} - a_{1})
\,+\,C_{2}\frac{ \delta}{2}\,(-2x|y| +\text{sgn}(y)a_{2})\,
|y|^{\frac{\delta}{2}-1}\\
&\simeq_{X_{2}}&-\,C_{2}\,y^{2}\ \rightarrow\ -\infty
\end{eqnarray*}  
as $\,r\rightarrow \infty\,$ in $\,X_{2}$.  
\end{proof}

\begin{prop}
Let $\,C_{3} >0\,$ and $\,\delta\in (0, 1/2)$.  \,Define 
\begin{equation*}
\varphi_{3}=C_{3}\left( \frac{x^{2}+y^{2}}{|y|^{3/2}}\right)^{\delta}\tag{LF3}
\end{equation*} 
on $\,X_{3}=\{x\leq -1\}\cap\{|y|\geq 1\}$.  \,Then 
$\,\varphi_{3}\in \mathcal{N}(\alpha,\kappa_{1}, \kappa_{2}, X_{3})\,$ 
for all $\,\alpha\in \mathbb{C}$ and all $\,\kappa_{1}\geq 0$, $\kappa_{2} >0$. 
\end{prop}
\begin{proof}
Smoothness of $\,\varphi_{3}\,$ is not a problem in this region as we are 
bounded away from $\,\mathbb{R}\,$ in $\,X_{3}$.  \,Clearly, 
$\,\varphi_{3}\geq 0\,$ and note that $\,\varphi_{3}\rightarrow \infty\,$ 
as $\,r\rightarrow \infty\,$ in $\,X_{3}$.  \,After dropping the 
$\,\delta(\delta-1)$-terms which are negative, we obtain: 
\begin{eqnarray}
\nonumber\tau M\varphi_{3}(x,y)&\leq& C_{3} 
\delta \left(\frac{x^{2}+y^{2}}{|y|^{3/2}}\right)^{\delta-1}
\Big[\kappa_{1}\frac{2}{|y|^{3/2}}+\kappa_{2} 
\left( \frac{15x^{2}}{4|y|^{7/2}}-\frac{1}{4|y|^{3/2}} \right)
+ \frac{x^{3}}{|y|^{3/2}}\\
\nonumber&\,&+ \,x|y|^{1/2}+\frac{2a_{1}x}{|y|^{3/2}}
+ \text{sgn}(y)a_{2}\left(-\frac{3x^{2}}{2|y|^{5/2}}
+\frac{1}{2|y|^{1/2}}\right)\Big]\\
\nonumber&\simeq_{X_{3}}&C_{3}\delta  
\left(\frac{x^{2}+y^{2}}{|y|^{3/2}}\right)^{\delta-1}
\left(\frac{x^{3}}{|y|^{3/2}}+x|y|^{1/2}\right)\\
\label{Mphi3}&=&\delta\, x\, \varphi_{3}\ \rightarrow\ -\infty
\end{eqnarray}  
as $\,r\rightarrow \infty\,$ in $\,X_{3}\,$ since $\,x\leq -1\,$ in $\,X_{3}$. 
\end{proof}

\begin{prop}
 Let $\,C_{4}>0, \,\eta > 1\,$ and $\,\delta\in (0, 1/2)$.  \,Define
 \begin{equation} 
 \varphi_{4}(x,y)=C_{4} \frac{|x|^{2\delta}+|y|^{2\delta}}{|y|^{\frac{3}{2}\delta}}\tag{LF4}
 \end{equation}
 on $$X_{4}\,=\,\{x\leq -1\}\,\cap\,\left\{\,\eta\hspace{0.03cm} 
\sqrt{\kappa_{2} 
\frac{_3}{^2}(\frac{_3}{^2}\delta +1)}\frac{_1}{^{\sqrt{|x|}}}\leq |y| 
\leq 2\,\right\}.$$  Then $\,\varphi_{4}\in \mathcal{N}(\alpha,\kappa_{1}, 
\kappa_{2}, X_{4})\,$ for all $\,\alpha \in \mathbb{C}\,$ and all 
$\,\kappa_{1} \geq 0$, $\kappa_{2} > 0$.  
 \end{prop}

 \begin{proof}
Note that $\,\varphi_{4}\,$ is smooth in $\,X_{4}\,$ since this region 
excludes both $\,x\,$ and $\,y\,$ axes.  Moreover, 
$\,\varphi_{4}\rightarrow \infty\,$ as $\,r \rightarrow \infty\,$ in 
$\,X_{4}\,$ since then $\,x\,$ must approach $\,\infty\,$ as 
$\,r\rightarrow \infty$, \,and $\,y\,$ is bounded above.  Dropping 
insignificant terms in the expression for $\,M\varphi_{4}$, \,we 
see that in $\,X_{4}$:  
 \begin{eqnarray*}
 \tau M\varphi_{4}(x,y)&\leq& C_{4}\Big{(}-\delta\frac{|x|^{2\delta+1}}
{|y|^{\frac{3}{2}\delta}}+\kappa_{2}\frac{_3}{^2}\delta\left(\frac{_3}{^2}
\delta +1\right) \frac{|x|^{2\delta}}{|y|^{\frac{3}{2}\delta+2}}+\delta |x| |y|^{\delta/2}\\
 &\,& -\,a_{1}2\delta \frac{|x|^{2\delta -1}}{|y|^{\frac{3}{2}\delta}}
+\text{sgn}(y)a_{2} \frac{_1}{^2}\delta|y|^{\frac{\delta}{2}-1}-\text{sgn}(y) 
a_{2}\frac{_3}{^2} \delta \frac{|x|^{2\delta}}{|y|^{\frac{3}{2}\delta +1}}
\Big{)}\\
 &\simeq_{X_{4}}&C_{4}\left(-\delta\frac{|x|^{2\delta+1}}{|y|^{\frac{_3}{^2}
\delta}}+\kappa_{2}\frac{_3}{^2}\delta\left(\frac{_3}{^2}\delta 
+1\right) \frac{|x|^{2\delta}}{|y|^{\frac{3}{2}\delta+2}}\right)\\
 &=&C_{4} \delta\frac{|x|^{2\delta +1}}{|y|^{\frac{3}{2}\delta}} 
\left( -1+\kappa_{2} \frac{3}{2} \left( \frac{3}{2}\delta 
+ 1\right)\frac{1}{|x||y|^{2}}\right)\\
 &\leq &-\,C_{4} \delta (1-1/\eta^{2}) \frac{|x|^{2\delta +1}}
{|y|^{\frac{3}{2}\delta}}\ \rightarrow\ -\infty
 \end{eqnarray*}   
 in $\,X_{4}\,$ as $\,r \rightarrow \infty $.  
 \end{proof}

\begin{prop}
Let $\,C_{5},\,\beta > 0\,$ and $\,E>0\,$ such that 
$\,2\kappa_{2} > E\beta$, \,let $\,\xi > 1$, \,and let 
\begin{equation}
\varphi_{5}(x,y)=C_{5}(E|x|^{\beta}-y^{2}|x|^{\beta+1})\tag{LF5}
\end{equation}  
be defined on $$X_{5}\,=\,\{x\leq -1\}\,\cap\,
\left\{ |y|\leq \frac{_1}{^\xi}\sqrt{\frac{_E}{^{|x|}}}\,\right\}.$$
Then $\,\varphi_{5}\in \mathcal{N}(\alpha,\kappa_{1},\kappa_{2}, X_{5})\,$ 
for all $\,\alpha \in \mathbb{C}\,$ and all $\,\kappa_{1}\geq 0$.  
\end{prop}

\begin{proof}
The fact that $\,\varphi_{5}\,$ is smooth in $\,X_{5}\,$ is clear as 
$\,x\leq - 1\,$ in $\,X_{5}$. \,Again by the choice of $\,X_{5}$, 
$\,\varphi_{5}\geq 0\,$ and $\,\varphi_{5}\rightarrow \infty\,$ as 
$\,r \rightarrow \infty\,$ in $\,X_{5}$. \,Dropping irrelevant terms 
in the expression for $\,M\varphi_{5}$, \,we see that in $\,X_{5}$:
\begin{eqnarray*}
\tau M\varphi_{5}(x,y)&\leq & C_{5}\left(\kappa_{1} E\beta (\beta -1) |x|^{\beta -2} -2\kappa_{2} |x|^{\beta +1} + E\beta |x|^{\beta +1}+(\beta +1) y^{4} |x|^{\beta}\right)\\
 &\,&+\,\,C_{5}\left(a_{1}(-E\beta |x|^{\beta -1} +(\beta +1) y^{2} |x|^{\beta})-2a_{2} y |x|^{\beta +1}\right)\\
 &\simeq_{X_{5}}& C_{5}\,(E\beta-2\kappa_{2})\,|x|^{\beta+1}\ \rightarrow\ -\infty
\end{eqnarray*}    
as $\,r \rightarrow \infty\,$ in $\,X_{5}$, \,as $\,|x|\,$ must approach 
$\,\infty\,$ when $\,r\rightarrow \infty\,$ in $\,X_{5}$.   
\end{proof}

%===================================================================
%------------------------------Picking parameters and varying diff. coeff. lemma-------------------
%=====================================================================

We now have our desired LFs.  It is not obvious, however, that the regions 
$\,X_{1},\,X_{2}, \ldots, X_{5}\,$ cover $\,\mathbb{R}^{2}\,$ except, possibly, 
a bounded region about the origin.  To assure that one has to show that 
$\,X_{4}\,$ and $\,X_{5}\,$ overlap.  In order to make this more tangible, 
we will choose some of the parameters given in the previous propositions.  
With the choices that follow, however, we first need a lemma that says that
varying the diffusion coefficients $\,(\kappa_{1},\kappa_{2})$ is permitted.  
This lemma will also be of crucial use later when we glue the LFs to form 
a globally defined $\,\Phi$.

\begin{lem}\label{LEM6}
Fix $\,\kappa_{2} >0\,$ and suppose that $\,\Phi \in \mathcal{N}(\alpha,
\kappa_{1},\kappa_{2},\mathbb{R}^{2})\,$ for all $\,\alpha \in \mathbb{C}\,$ 
and all $\,\kappa_{1}\geq 0$.  \,Then for every $\,\iota_{2}>0$, \,we can 
find $\,\Psi \in \mathcal{N}(\alpha,\iota_{1},\iota_{2},\mathbb{R}^{2})\,$ 
for all $\,\alpha \in \mathbb{C}\,$ and all $\,\iota_{1} \geq 0$.      
\end{lem}    

\noindent For the proof of Lemma \ref{LEM6}, we temporarily use the 
notation $\,M^{\alpha}_{(\kappa_{1}, \kappa_{2})}\,$ for the generator $\,M\,$ 
given by \eqref{genw}.  

\begin{proof}
Let $\,\eta > 0\,$ be such that $\,\eta^{3} \iota_{2}=\kappa_{2}$. \,Define 
$\,\Psi(x,y)=\Phi(\eta x, \eta y)$.  \,For the function $\,\Psi$, 
\,smoothness and properties \eqref{I} and \eqref{II} are immediate.  
Let $\,s=\eta x\,$ and $\,t=\eta y$.  \,Then by the chain rule 
\begin{eqnarray*}
\tau M^{\alpha}_{(\iota_{1}, \iota_{2})}\Psi(x,y)&=&\frac{_\tau}{^\eta}\,
M^{\alpha \eta^{2}}_{(\iota_{1} \eta^{3},\iota_{2} \eta^{3})} \Phi(s,t) \\
&=&\frac{_\tau}{^\eta}\,M^{\alpha \eta^{2}}_{(\iota_{1} \eta^{3},\kappa_{2})} 
\Phi(s,t)\ \rightarrow\ -\infty,  
\end{eqnarray*} 
as $\,r=\sqrt{x^{2}+y^{2}}\rightarrow \infty$.   
\end{proof}

\noindent By Lemma \ref{LEM6}, it is enough to find a function 
$\,\Phi\in \mathcal{N}(\alpha,\kappa_{1},\kappa_{2},\mathbb{R}^{2})\,$ 
for some fixed $\,\kappa_{2} > 0\,$ for all $\,\alpha \in\mathbb{C}\,$ 
and all $\,\kappa_{1}\geq 0$.  \,All of the $\,\varphi_{i}\,$ satisfy these 
criteria.  In fact, $\,\varphi_{i}\,$ for $\,i=1,2,3,4\,$ work more 
generally.  The reason that $\,\varphi_{5}\,$ only works for 
$\,2\kappa_{2}>E\beta\,$ is due to the fact that when 
$\,\alpha=\kappa_{1}=\kappa_{2}=0$, \,the solution to \eqref{sdew} has 
an explosive trajectory along the negative real axis.  
\vskip 0.1cm

Now we choose some parameters. \,Let $\,E=5\,$  and $\,\xi=\sqrt{5}/2$, \,so 
as to make $$X_{5}\,=\,\{ x\leq -1\,\cap\,\{|y||x|^{1/2}\leq 2\}.$$  
Let $\,\beta =\frac{11}{4}\delta$, $\,\delta =\frac{1}{7}\kappa_{2}$, \,and 
$\,\kappa_{2}\in(0,1)$.  \,Then $\,E\beta =\frac{55}{28}
\kappa_{2} < 2\kappa_{2}\,$ and $\,\delta\in (0, \frac{1}{2})$.  \,Hence 
for all $\,i=1,2,\ldots, 5$, $\,\varphi_{i}\,$ is a LF in the region 
$\,X_{i}$.  \,Decrease $\,\kappa_{2}> 0\,$ so that 
$$X_{4}\,\supset\,\{x \leq -1\}\,\cap\,\{1 \leq |y||x|^{1/2}\leq 2\}.$$  
One can easily check that $X_{4}$ and $X_{5}$ overlap in such a way that we have covered all of $\mathbb{R}^{2}$ with $X_{1}, X_{2}, \ldots, X_{5}$ except a bounded region about the origin.  
\vskip 0.1cm

We fix $\,\kappa_{2}>0\,$ sufficiently small (this will be made precise later).
We will construct a function $\,\Phi\in \mathcal{N}(\alpha,\kappa_{1},
\kappa_{2},\mathbb{R}^{2})\,$ for all $\,\alpha \in \mathbb{C}\,$ and all 
$\,\kappa_{1}\geq 0$.  \,The idea is as follows.  Note that 
$\,\varphi_{1}\,$ and $\,\varphi_{2}\,$ are LFs in the region 
$\,X_{1}\cap X_{2}$.  \,We shall define a nonnegative smooth auxiliary 
function $\,\zeta(x)\in C^{\infty}(\mathbb{R})\,$ such that 
$\,\zeta(x)=0\,$ for $\,x\geq 2\,$ and $\,\zeta(x)=1\,$ for $\,x\leq 1$,  
\,satisfying some additional properties.  We will then show that the 
combination $$(1-\zeta) \varphi_{1}+\zeta \varphi_{2}$$ is a LF 
in the larger region $\,X_{1} \cup X_{2}$.  \,Proceeding inductively this way, 
we shall construct a LF in all of $\,\mathbb{R}^{2}$.  
\vskip 0.1cm

%==========================================================================
%-------------------------------------AUX FUNCTIONS DEFINED----------------------------------------------------
%==========================================================================

Let us first define some auxiliary functions needed to construct such 
a $\,\Phi$. \,Let $\,\zeta:\mathbb{R} \rightarrow \mathbb{R}_{+}\,$ 
be a $C^{\infty}$ function such that  
\begin{align*}
\zeta(x)&=
\begin{cases}
\,1 & \text{ for \,} x\leq 1\,,\\
\,0 & \text{ for \,} x\geq 2\,,
\end{cases}
\end{align*}
and $\,\zeta'(x)<0\,$ for all $\,x\in (1, 2)$.    
\,We define the smooth function $\,\mu:\mathbb{R}\rightarrow \mathbb{R}_{+}\,$ 
as the horizontal shift of $\,\zeta$, \,three units to the left, 
i.e., $$\mu(x)=\zeta(x+3) \text{ \ for } x\in \mathbb{R}.$$  Let 
$$\nu(x,y)=\begin{cases}\,\zeta(|y|) & \text{\ for \,} x\leq -2\,,\\ \,0& 
\text{ for \,}|y|\geq 2\,,\\ \,0 & \text{ for \,} x >-1\,,\end{cases}$$ 
and assume that $\nu$ is $C^{\infty}$ outside of the ball $\,B_{4}$.  
\,Let $\ q:(-\infty, -1]\times \mathbb{R}\rightarrow \mathbb{R}\ $ be 
defined by $$q(x,y)=\begin{cases}\,1 & \text{ if\quad} |x|^{1/2}|y|\geq 2\,, 
\\ \,|x|^{1/2}|y|-1 &\text{ if\quad} 1< |x|^{1/2} |y|< 2\,,\\ \,0 &
\text{ if\quad} |x|^{1/2}|y|\leq 1\,,\end{cases}$$ and 
$$r(t)=\begin{cases}\,\exp\left(-\frac{1}{1-(2t-1)^{2}}\right) & 
\text{ if\quad} 0< t < 1\,,\\ \,0 &\text{ otherwise}\,.\end{cases}$$  
Let $$s(x)=\frac{1}{N}\int_{-\infty}^{x} r(t)\,dt\,,$$ 
where $\,N=\int_{\mathbb{R}} r(t) dt$.  \,Now define a function on 
$\,\mathbb{R}^{2}\,$ by $$\rho(x,y)=\begin{cases}\,s(q(x,y))& \text{ if\quad} 
x\leq -1\,,\\ \,1 &\text{ if\quad} |y|\geq 3\,,\\ \,1 &\text{ if\quad} 
x\geq -1/2\,.\end{cases}$$  Clearly, $\,\rho\,$ is $C^{\infty}(\mathbb{R}^{2})$ 
outside of $\,B_{4}$.  
%==========================================================================
%------------------------------------ little phi and big phi defined-------------------------------------------------------
%==========================================================================

Let $\,r(\delta)=\max\left(4, \sqrt{2^{4/\delta}+4} \right)$. \.Define 
$\,\varphi\,$ for $\,x^{2}+y^{2}\geq r^{2}(\delta)\,$ by  
$$\varphi(x,y)=\begin{cases}\,\varphi_{1} & \text{ if\quad} x\geq 2\,, \\
\,\zeta \varphi_{2} + (1-\zeta) \varphi_{1} &\text{ if\quad} 1< x< 2\,, 
\\ \,\varphi_{2} &\text{ if\quad} -1\leq x \leq 1\,, \\ 
\,\mu\varphi_{3}+(1-\mu)\varphi_{2} &\text{ if\quad} -2< x< -1\,, \\
\,\varphi_{3} &\text{ if\quad}  x\leq -2,\ |y|\geq 2 \\
\,\nu \varphi_{4}+ (1-\nu) \varphi_{3} &\text{ if\quad} x\leq -2,\ 
1<|y|< 2\,, \\ \,\varphi_{4} &\text{ if\quad} x\leq -2,\ |y|\leq 1,\ 
|x|^{1/2}|y|\geq 2\,,\\ \,\rho \varphi_{4} +(1-\rho) \varphi_{5}& 
\text{ if\quad} x\leq -2,\ 1<|x|^{1/2}|y|< 2\,,\\
\,\varphi_{5} &\text{ if\quad} x\leq -2,\ |x|^{1/2} |y|\leq 1\,,
\end{cases}$$ and now 

$$\Phi(x,y)=\begin{cases} \,\varphi(x,y) &\text{ if\quad} 
x^{2}+y^{2}\geq B_{r(\delta)}\,, \\ \text{arbitrary positive and smooth } 
&\text{ if\quad} x^{2}+y^{2} < B_{r(\delta)}\,. \end{cases}$$  
It is easy to see $\,\Phi\,$ can be chosen to be nonnegative and 
$C^{\infty}(\mathbb{R}^{2})$.  With the aid of Lemma \ref{LEM6}, 
the following lemma implies Theorem \ref{TH1} in the $\,d=2\,$ case.  

%==========================================================================
%---------------------------------------------patching lemma---------------------------------------------------------------
%==========================================================================

\begin{lem}\label{lem217}
For $\,\kappa_{2}\,$ sufficiently small, $\,\Phi\in 
\mathcal{N}(\alpha,\kappa_{1},\kappa_{2},\mathbb{R}^{2})\,$ for all 
$\alpha\in \mathbb{C}$ and all $\kappa_{1}\geq 0$.  
\end{lem}

\begin{proof}
Clearly, $\,\Phi\,$ is smooth and satisfies properties I and II.  \,Since 
$\,M\varphi_{i}\rightarrow -\infty\,$ as $\,r\rightarrow \infty\,$ in 
$\,X_{i}\,$ for each $\,i$, \,all we must verify is that 
$\,M\Phi\rightarrow -\infty\,$ as $\,r\rightarrow \infty\,$ in the 
overlapping regions.  Let us recall the choices that have already been 
made: 
\begin{align*}
E=5\,, \text{\quad} \xi = \frac{_{\sqrt{5}}}{^2}\,, 
\text{\quad}\beta= \frac{_{11}}{^4}\delta\,,\text{\quad} 
\delta =\frac{_{\kappa_{2}}}{^7}\,,
\end{align*}
and note that $\,\kappa_{2}\in (0,1)\,$ was chosen such that 
$$X_{4}\,\supset\,\{ x\leq -1\}\,\cap\,\{ 1\leq |y||x|^{1/2} \leq 2\}.$$
Pick $\,C_{5}>C_{4}=C_{3}>C_{2}>C_{1}$.  
\,Consider first $$\psi_{1}:=\zeta\varphi_{2}+(1-\zeta) \varphi _{1}$$ 
defined in the region  $$Y_{1}\,=\,\{1<x<2\}\,\cap\,B_{r(\delta)}^{c}.$$ 
We have 
\begin{eqnarray}
\nonumber \tau M\psi_{1}&=&\zeta\,\tau M \varphi_{2}
+(1-\zeta) \tau M \varphi_{1} + (y^{2}-x^{2}
+a_{1})\zeta' (\varphi_{2}-\varphi_{1})\\
\nonumber &\,&+\,\,\kappa_{1}(\zeta'' (\varphi_{2}-\varphi_{1})+ 
2\zeta'(\partial_{x} \varphi_{2}-\partial_{x} \varphi_{1}))\\
\label{psi1}&\simeq_{Y_{1}}& -\,\zeta\,C_{2} y^{2}
-(1-\zeta) \frac{_{C_{1} \delta}}{^2}x(x^{2}+y^{2})^{\delta/4}
+(y^{2}-x^{2}+a_{1})\zeta'(C_{2}-C_{1}) |y|^{\delta/2}\\
\nonumber &\,&+\,\,\kappa_{1}(\zeta''(C_{2}-C_{1}) 
|y|^{\delta/2}-2\zeta' C_{2})\,.
\end{eqnarray}  
Note that if $\,x\,$ is bounded away from $\,1\,$ and $\,2\,$ in 
$\,(1,2)$, \,the dominant term above is $\,y^{2}\zeta'(x)(C_{2}-C_{1}) 
|y|^{\delta/2}\rightarrow -\infty$.  \,Note also that as 
$\,x\rightarrow 1\,$ or $\,x\rightarrow 2$, $\,\zeta',\,\zeta''
\rightarrow 0$. ,But, the first two terms in the expression above decay 
to $\,-\infty\,$ at least as fast as $\,-C|y|^{\delta/2}$, 
\,where $\,C>0\,$ is a constant independent of $\,\zeta\,$ and $\,x$.  
\,Thus we may choose $\,\epsilon > 0\,$ so that whenever $\,x\in (1, 2) 
\setminus (1+\epsilon, 2-\epsilon)$, $\,M \psi_{1}\,$ decays at least as 
fast as $\,-D|y|^{\delta/2}$, where $\,D\,$ is some positive constant.

\noindent We now consider $$\psi_{2}:=\mu \varphi_{3}+(1-\mu) \varphi_{2}$$ 
in the region $$Y_{2}\,=\,\{ -2< x<-1\}\,\cap\,B_{r(\delta)}^{c}.$$ 
We have
\begin{eqnarray*}
\tau M\psi_{2}&=&\mu\,\tau  M \varphi_{3}+(1-\mu) \tau M \varphi_{2} 
+ (y^{2}-x^{2}+a_{1})\mu'(\varphi_{3}-\varphi_{2})\\
&\,&+\,\kappa_{1}(\mu''(\varphi_{3}-\varphi_{2})+ 2\mu'(\partial_{x} 
\varphi_{3}-\partial_{x} \varphi_{2})\\
&\simeq_{Y_{2}}& \mu\,\delta x\varphi_{3} -(1-\mu) C_{2} y^{2}
+(y^{2}-x^{2}+a_{1}) \mu'(C_{3}-C_{2}) |y|^{\delta/2}\\
&\,&+\,\kappa_{1}(\mu''(C_{3}-C_{2})|y|^{\delta/2}+2\mu'C_{2})\\
&\simeq_{Y_{2}}&\mu\,C_{3}x|y|^{\delta/2}-(1-\mu) C_{2} y^{2}
+(y^{2}-x^{2}+a_{1}) \mu'(C_{3}-C_{2}) |y|^{\delta/2}\\
&\,&+\,\kappa_{1}(\mu''(C_{3}-C_{2})|y|^{\delta/2}+2\mu'C_{2}).
\end{eqnarray*}
Note that, for the very same reasons as in the case of $\,\psi_{1}$, 
$\,M\psi_{2}\rightarrow -\infty\,$ as $\,r \rightarrow \infty\,$ 
in $\,Y_{2}$.

\noindent Let $$\psi_{3}:=\nu \varphi_{4}+(1-\nu) \varphi_{3}$$ in the region 
$$Y_{3}\,=\,\{x \leq -2 \}\,\cap\,\left\{1 <|y|< 2\right\}\,\cap\,
B_{r(\delta)}^{c}.$$  Thus
\begin{eqnarray*}
 \tau M\psi_{3}&=& \nu\,\tau M \varphi_{4}+(1-\nu) \tau M \varphi_{3} 
+ (-2x|y|+\text{sgn}(y)a_{2})(\partial_{|y|}\nu) (\varphi_{4}-\varphi_{3})\\ 
&\,&+\, \kappa_{2} (\partial^2_{|y|}\nu)(\varphi_{4}-\varphi_{3})
+2\kappa_{2}(\partial_{|y|}\nu)\partial_{|y|}(\varphi_{4}-\varphi_{3})\\
&\simeq_{Y_{3}}&  \nu\,\tau M \varphi_{4}+(1-\nu) \tau M \varphi_{3}\ 
\rightarrow\ -\infty\,.
\end{eqnarray*}
This is true since we chose $\,C_{3}=C_{4}$.  \,Hence the first order term 
$$(-2x|y|+\text{sgn}(y)a_{2})(\partial_{|y|}\nu)(\varphi_{4}-\varphi_{3})$$
approaches infinity at worst as fast as $\,C|x|$, \,where $\,C\,$ is a 
positive constant. The second order term $\,\kappa_{2}(\partial_{|y|}^2\nu)
(\varphi_{4}-\varphi_{3})\,$ is a bounded function in this region.  
Moreover the term $\,2\kappa_{2}(\partial_{|y|}\nu)\partial_{|y|}
(\varphi_{4}-\varphi_{3})\,$ at worst approaches infinity as fast as 
$\,D|x|^{2\delta}$, \,where $\,D\,$ is some positive constant.  But, both 
$\,M\varphi_{3}\,$ and $\,M\varphi_{4}\,$ approach negative infinity 
at least as fast as $\,-D|x|^{2\delta +1}$, \,where $\,D\,$ is another 
positive constant.  This gives the desired result.

\noindent Let $$\psi_{4}:=\rho \varphi_{4}+(1-\rho)\varphi_{5}, $$ 
in the region $$Y_{4}\,=\,\{1<|x|^{1/2}|y|< 2\}\,\cap\,\{x\leq -2\}
\,\cap\,B_{r(\delta)}^{c}.$$  Note that 
\begin{eqnarray*}
\tau M\psi_{4}&=& s(q)\tau M\varphi_{4}+(1-s(q)) \tau M\varphi_{5}
+\partial_{x}(s(q))(\varphi_{4}-\varphi_{5})(y^{2}-x^{2}+a_{1})\\
&\,&+\,\partial_{|y|}(s(q))(\varphi_{4}-\varphi_{5})(-2x|y|
+\text{sgn}(y)a_{2})+\kappa_{1}\partial_{x}^2 (s(q))(\varphi_{4}-\varphi_{5})\\
&\,&+\,2\kappa_{1} \partial_{x}(s(q))\partial_{x}(\varphi_{4}-\varphi_{5})
+\kappa_{2} \partial_{|y|}^2 (s(q))(\varphi_{4}-\varphi_{5})\\
&\,&+\, 2\kappa_{2}\partial_{|y|}(s(q))\partial_{|y|}(\varphi_{4}
-\varphi_{5})\,.
\end{eqnarray*}
Note that in the expression above, we may drop the 
$\,\kappa_{1}\partial_{x}^2(s(q))(\varphi_{4}-\varphi_{5})\,$ 
and $\,2\kappa_{1} \partial_{x}(s(q))\partial_{x}(\varphi_{4}-\varphi_{5})\,$ 
terms, as they are asymptotically less than other terms.  Dropping other 
obviously insignificant terms, we obtain:
\begin{eqnarray*}
\tau M\psi_{4}&\simeq_{Y_{4}}& s(q)\tau M\varphi_{4}+(1-s(q)) \tau M\varphi_{5}-\partial_{x}(s(q))(\varphi_{4}-\varphi_{5})x^{2}\\
&\,&-\,\partial_{|y|}(s(q))(\varphi_{4}-\varphi_{5})2x|y|+\kappa_{2} 
\partial_{|y|}^2(s(q))(\varphi_{4}-\varphi_{5})\\
&\,& +\, 2\kappa_{2}\partial_{|y|}(s(q)) \partial_{|y|}(\varphi_{4}-
\varphi_{5})\\
&\leq&  F(x,y)\,,
\end{eqnarray*} 
where $F$ is a smooth function satisfying:
\begin{eqnarray}\label{Fest}
F &\simeq_{Y_{4}}& -\,s(q)C_{4}(1-1/\eta^{2}) \frac{\delta}{2^{3\delta/2}} 
|x|^{\frac{11}{4}\delta + 1} -(1-s(q))\frac{C_{5}\delta}{4} 
|x|^{\frac{11}{4}\delta + 1} \\
\nonumber&\,&-\, \partial_{x}(s(q))(\varphi_{4}-\varphi_{5})x^{2}
-\partial_{|y|}(s(q))(\varphi_{4}-\varphi_{5})2x|y|\\
\nonumber&\,&+\,\kappa_{2} \partial_{|y|}^2 (s(q))(\varphi_{4}-\varphi_{5})
+2\kappa_{2}\partial_{|y|}(s(q))\partial_{|y|}(\varphi_{4}-\varphi_{5})\,.
\end{eqnarray}    
\noindent By the choice of $\,C_{5}>C_{4}$, \,it is easy to see that for 
large $\,|x|\,$ in this region, there are constants $\,C,\,D>0\,$ such that 
$$\left|\kappa_{2} \partial_{|y|}^2(s(q))(\varphi_{4}-\varphi_{5})
+ 2\kappa_{2}\partial_{|y|}(s(q))\partial_{|y|}(\varphi_{4}-\varphi_{5})
\right|\,\leq\,C\kappa_{2}r(q)\frac{|x|^{\frac{11}{4}\delta +1}}{g(x,y)^{2}}$$ 
and 
$$-\partial_{x}(s(q))(\varphi_{4}-\varphi_{5})x^{2}-\partial_{|y|}(s(q))
(\varphi_{4}-\varphi_{5})2x|y|\,\leq\,-Dr(q) |x|^{\frac{11}{4}\delta +1},$$ 
where the constants $\,C,\,D\,$ and the function $\,r(q)\,$ are independent 
of $\,\kappa_{2}$.  $\,r(q)\,$ is a function that goes to zero faster 
than any power of the function $$g(x,y):=|(|x|^{1/2}|y|-1)(|x|^{1/2}|y|-2)|$$
as  $\,|x|^{1/2}|y|\rightarrow 1 \text{ or }2$.  \,But note that, for all 
$\,\epsilon > 0$, \,we may choose $\,\kappa_{2}\,$ so small so that  
$$D > \frac{C\kappa_{2}}{g(x,y)^{2}}$$ 
for $\,1+\epsilon \leq |x|^{1/2}|y|\leq 2-\epsilon$.  \,From the first 
two terms in \eqref{Fest}, we obtain at least $\,-C' \delta\,
x^{\frac{11}{4}\delta+1}\,$ decay for some $\,C'>0\,$ independent of 
$\,\kappa_{2}\,$ for all $\,1\leq x^{1/2} |y| \leq 2$.  \,But since every 
other term in the expression goes to zero faster than every power of $\,g\,$ 
as $\,|x|^{1/2}|y|\rightarrow 1\,$ or $\,2$, \,we can choose an 
$\,\epsilon > 0\,$ so small as above so that $\,M\psi_{4} 
\rightarrow -\infty\,$ for all $\,1< |x|^{1/2}|y|<2$.  \,This completes 
the proof.    
\end{proof}

%===============================================================
%===============================================================
%------------------------NEW SUBSECTION: MORE ON $d\geq 3$ CASE--------------------
%===============================================================
%===============================================================

\subsection{$d\geq 3\,$ case}
\label{subsec:d=3}

\ 

\vskip 0.1cm

\noindent Here we shall complete the proof of Theorems \ref{TH1} 
for $\,d\geq 3$.  \,Recall that for $\,d=2\,$ and 
$\,\epsilon >0\,$ sufficiently small, we defined $\,\Phi_{2}(x,y):=
\Phi(x,y) \in C^{\infty}(Q_2)\,$ that satisfies \eqref{I}, \eqref{II}, 
and \eqref{IV} for all $\,\alpha \in \mathbb{C}$, $\,\kappa_{1} \geq 0$, 
\,and $\,\kappa_{2}\in (0, \epsilon)$.  \,Lemma \ref{LEM6} implied then
that for $\,\kappa_{2}>0\,$ arbitrary, 
the function $\,\Phi_{2, \eta}(x,y):=\Phi_{2}(\eta x, \eta y)
\in C^{\infty}(Q_2)\,$ satisfied \eqref{I}, \eqref{II}, and \eqref{IV} 
for all $\,\alpha$, \,provided that $\,\epsilon \eta^{-3} = 2\kappa_{2}$.  
\,Let is fix $\,\kappa_{2} >0$.  \,For $\,d\geq 3$, \,we shall define
\begin{equation}
\Phi_{d, \eta}:= \Phi_{2, \eta} + \log(1 + \log^{2}(\eta y/2))\,.
\end{equation}      
\begin{lem}
For fixed $\,\kappa_{2}>0$, $\,\Phi_{d, \eta}$ is a smooth function 
on $\,Q_d\,$ that satisfies \eqref{I}, \eqref{II}, and \eqref{IV}.
\end{lem} 
\begin{proof}
By definition, $\,\Phi_{d, \eta}\,$ is smooth and nonnegative in 
$\,Q_d=\mathbb{H}_{+}$.  \,Clearly, $\,\Phi_{d, \eta}\rightarrow \infty$ 
as $\,(x,y)\rightarrow \partial Q_d=\{(x,y) \in \mathbb{R}^{2} : y=0\}
\cup \{ \infty\}$, $\,(x,y)\in Q_d$.  \,Thus we must verify property 
\eqref{IV}.  To this end, note that:
\begin{eqnarray*}
\tau M_{d}\Phi_{d, \eta}&=& \tau M_{d} \Phi_{2, \eta} + \tau 
M_{d} (\log(1+ \log^{2}(\eta y/2)))\\
&=& \tau M_{2} \Phi_{2, \eta} + \tau \frac{b(d-2)}{y} 
\partial_y\Phi_{2, \eta}+\left(-2xy^2+a_2y+\tau b(d-2) -\kappa_{2} \right)
\frac{2\log(\eta y/2)}{y^{2}(1+ \log^{2}(\eta y/2))}\\
&\,&  +\, \frac{2\kappa_{2}(1- \log^{2}(\eta y/2))}{y^{2}(1+\log^{2}(
\eta y/2))^{2}}.
\end{eqnarray*} 
\noindent \textbf{Case 1.}  Suppose first that $\,y \geq 2\eta^{-1}$.  
\,It is easy to check that there exist positive constants 
$\,K_{1},\,K_{2}>0\,$ such that 
\begin{eqnarray*}
\qquad \tau \frac{b(d-2)}{y} \partial_y \Phi_{2, \eta}\,\leq\,K_{1}, 
\quad\left(a_2y+\tau b(d-2) -\kappa_{2} \right)\frac{2\log(\eta y/2)}{y^{2}
(1+ \log^{2}(\eta y/2))}
+ \frac{2\kappa_{2}(1-\log^{2}(\eta y/2))}{y^{2}(1+\log^{2}(\eta y/2))^{2}}\,
\leq\,K_{2}\,.
\end{eqnarray*}
If also $\,x>-2\eta^{-1}\,$ then
\qq
-2x\frac{\log(\eta y/2)}{1+\log^2(\eta y/2)}\,\leq\,K_3
\qqq 
for a positive constant $\,K_3$, \,whereas for $\,x\leq-2\eta^{-1}$,
\qq
-2x\frac{\log(\eta y/2)}{1+\log^2(\eta y/2)}\,\leq\,-K_4x\,.
\qqq 
Since $\,M_{2} \Phi_{2, \eta}\rightarrow -\infty\,$ as $\,(x,y)
\rightarrow\partial Q_d\,$ and $\,y\geq 2\eta^{-1}$, \,and, besides, 
if $\,x\leq-2\eta^{-1}\,$
and $\,y\geq 2\eta^{-1}\,$ then $\,\Phi_{2,\eta}(x,y)\,$ is equal to 
the rescaled function $\,\varphi_3\,$ so that, by \eqref{Mphi3},
$\,M_2\Phi_{2,\eta}\leq K_5x(x^2+y^2)^{\delta/4}\,$ 
for some $\,K_5>0$, \,we infer that 
$\,M_d\Phi_{d,\eta}\rightarrow-\infty\,$ 
as $\,(x,y)\rightarrow\partial Q_d\,$ and $\,y\geq 2\eta^{-1}$. 
\vspace{0.1in}
 
\noindent \textbf{Case 2.}  Suppose now that $\,0<y<2\eta^{-1}$.  \,If 
$\,|x|<2\eta^{-1}$, \,then $\,(x,y) \rightarrow \partial Q_d\,$ 
if and only if 
$\,y \downarrow 0$.  \,Since $\,\Phi_{2, \eta}\,$ is smooth on 
$\,\mathbb{R}^{2}$, \,there exists a constant $\,K_{6}>0\,$ such that 
\begin{equation*}
\tau M_{2} \Phi_{2, \eta} + \frac{\tau b(d-2)}{y} 
\frac{\partial \Phi_{2, \eta}}{\partial y}\,\leq\,K_{6}\frac{1}{y}
\end{equation*}  
for $\,(x,y)\in (-2\eta^{-1},2\eta^{-1})\times (0,2\eta^{-1})$.  
\,Hence, recalling the assumption $\,\tau b(d-2)\geq\kappa_2$, \,we 
have on this rectangle:  
\begin{eqnarray*}
\tau M_{d} \Phi_{d, \eta}& \leq & K_{6}\frac{1}{y} + \left(-2xy^2+a_2y+
\tau b(d-2) -\kappa_{2} \right)\frac{2\log(\eta y/2)}
{y^{2}(1+\log^{2}(\eta y/2))}
+ \frac{2\kappa_{2}(1-\log^{2}(\eta y/2))}{y^{2}(1+\log^{2}(\eta y/2))^{2}}\\
&\leq&K_{7}\frac{1}{y} +  \frac{2\kappa_{2}(1-\log^{2}(\eta y/2))}{y^{2}
(1+\log^{2}(\eta y/2))^{2}}\ \rightarrow\ -\infty
\end{eqnarray*}
as $\,y\downarrow 0$.  

\noindent If $\,x\geq 2\eta^{-1}$, \,then $\,\Phi_{2, \eta}\,$
is equal to the rescaled function $\,\varphi_{1}$.  
\,We see that by \eqref{varphi1} there exist constants 
$\,K_{8},\,K_{9},\,K_{10},\,K_{11}>0\,$ such that
\begin{eqnarray*}
\tau M_{d} \Phi_{d, \eta} &\leq &-K_{8}\,x(x^{2}+y^{2})^{\delta/4}  
+ K_9(x^2+y^2)^{\delta/4-1}\\
& \,& \,+\left(-2xy^2+a_2y+\tau b(d-2) 
-\kappa_{2} \right)\frac{2\log(\eta y/2)}
{y^{2}(1+ \log^{2}(\eta y/2))}+ 
\frac{2\kappa_{2}(1-\log^{2}(\eta y/2))}{y^{2}(1+\log^{2}(\eta y/2))^{2}}\\
&\leq&-K_{8}\,x(x^{2}+y^{2})^{\delta/4}+K_{10}x+K_{11}\frac{1}{y}+
\frac{2\kappa_{2}(1-\log^{2}(\eta y/2))}{y^{2}(1+\log^{2}(\eta y/2))^{2}}\,.
\end{eqnarray*}       
Note that as $\,(x,y) \rightarrow \partial Q_d\,$ in this region then
$\,x\rightarrow \infty\,$ or $\,y\downarrow 0$.  \,It is thus easy to see 
that $\tau M_{d} \Phi_{d, \eta} \rightarrow -\infty$.    
 
 \noindent If $\,x<-2\eta^{-1}$, \,then it is easy to check that 
$\,\partial_{y}\Phi_{2, \eta}\,$ is bounded above by the choice 
of $\,C_3=C_4\,$ and $C_{5}>C_{4}$.  \,Then, for some $\,K_{12}>0$,
\begin{eqnarray*}
\tau M_{d} \Phi_{d, \eta} &\leq &\tau M_2\Phi_{2,\eta}+K_{12}\frac{1}{y}+
\frac{2\kappa_{2}(1-\log^{2}(\eta y/2))}{y^{2}(1+\log^{2}(\eta y/2))^{2}}
\end{eqnarray*}
so that that $\,\tau M_{d} \Phi_{d, \eta} \rightarrow -\infty\,$ 
as $\,(x,y) \rightarrow \partial Q_d\,$ in this region.      
\end{proof}

%==============================================================
%------------------------------Top Lyapunov Exponent------------------------------------------
%==============================================================

\nsection{\bf TOP LYAPUNOV EXPONENT}
\label{sec:lyap}

\noindent The Lyapunov exponent $\,\lambda\,$ for the dispersion process 
$\,{\bm p}(t)=({\bm\rho}(t),{\bm\chi}(t))\,$ is the asymptotic rate 
of growth in time of the logarithm of the length 
$\,\sqrt{{\bm\rho}^2+{\bm\chi}^2}$.
\,Suppose that the process starts at $\,t=0\,$ from 
$\,{\bm p}_0=({\bm\rho}_0,{\bm\chi}_0)\not=0$. \,Anticipating the
existence of the limit below, we shall define:
\qq
\lambda&=&\lim\limits_{T\to\infty}\ \frac{_1}{^T}\,
\Big\langle\big(\ln{\sqrt{{\bm\rho}^2(T)+{\bm\chi}^2(T)}}
\,-\,\ln{\sqrt{{\bm\rho}^2_0+{\bm\chi}^2_0)}}\big)\Big\rangle\cr
&=&\lim\limits_{T\to\infty}\ \frac{_1}{^T}\int\limits_0^T
\frac{d}{dt}\Big\langle
\ln{\sqrt{{\bm\rho}^2(t)+{\bm\chi}^2(t)}}\Big\rangle\,dt\cr
&=&\lim\limits_{T\to\infty}\ \frac{_1}{^T}\int\limits_0^T\Big\langle
L\,\ln{\sqrt{{\bm\rho}^2(t)+{\bm\chi}^2(t)}}\Big\rangle\,dt\cr
&=&\lim\limits_{T\to\infty}\ \frac{_1}{^T}\int\limits_0^T 
\Big(\int(L\,\ln{\sqrt{{\bm\rho}^2+{\bm\chi}^2}})\,\,
 P_t({\bm p}_0,d{\bm p})\Big)\,dt\,,
\qqq
where $\,L\,$ is the generator of the process $\,{\bm p}(t)\,$ 
given by Eq.\,(\ref{genL}).
\,Note that the function $\,L\,\ln{\sqrt{{\bm\rho}^2+{\bm\chi^2}}}\,$ 
is smooth on $\,\mathbb{R}^{2d}\setminus\{0\}\,$ and homogeneous
of degree zero. \,It may be viewed as a function $\,f_0({\bm\pi})\,$
on $\,S^{2d-1}\,$ that, besides, is $\,SO(d)$-invariant. \,We may 
then rewrite the definition of $\,\lambda\,$
as
\qq
\lambda\ =\ \lim\limits_{T\to\infty}\ \frac{_1}{^T}\int\limits_0^T 
\Big(\int f_0({\bm\pi})\, P_t({\bm\pi}_0;d{\bm\pi})\Big)\,dt\,.
\label{explam}
\qqq
Now, the existence of the limit follows from the fact that the
Cesaro means of the transition probabilities $\,P_t({\bm\pi}_0;d{\bm\pi})\,$
tend in weak topology to the unique invariant probability measure 
$\,\mu(d{\bm\pi})$.
\,Hence 
\qq
\lambda\ =\ \int f_0({\bm\pi})\,\mu(d{\bm\pi})
\label{lambdaf0}
\qqq
and is independent of $\,{\bm p}_0$.
\,The crucial input that allows to make the latter formula more
explicit is the formula
\qq
L\,\ln{\sqrt{{\bm\rho}^2}}\ =\ \frac{1}{\tau}\,\frac{{\bm\rho}
\cdot{\bm\chi}}{{\bm\rho}^2}\ =\ \frac{x}{\tau}\,.
\qqq
It implies that
\qq
f_0({\bm\pi})\ =\ L\,\ln{\sqrt{{\bm\rho}^2+{\bm\chi}^2}}
&=&\frac{x}{\tau}\, 
+\,\frac{1}{2}\,L\,\ln{\Big(1+\frac{{\bm\chi}^2}{{\bm\rho}^2}\Big)}\cr\cr
&=&\frac{x}{\tau}\,+\,\frac{1}{2}\,L
\begin{cases}\,\hbox to 3.5cm{$\ln{(1+x^2)}$\hfill}{\rm for}
\quad d=1\cr
\,\hbox to 3.5cm{$\ln{(1+x^2+y^2)}$\hfill}{\rm for}\quad d\geq2
\end{cases}
\label{f0}
\qqq
in terms of the $\,SO(d)\,$ invariants with $\,L\,$ given by
explicit formulae (\ref{Lonf1}), (\ref{Lonf2}) or (\ref{Lonf3}). 
\,If the functions $\,\ln(1+x^2)\,$ in $\,d=1\,$ and 
$\,\ln{(1+x^2+y^2)}\,$ in $\,d=2\,$ that are homogeneous of 
degree zero on $\,\mathbb{R}^{2d}\setminus\{0\}$ \,were smooth, 
then their contributions 
to the expectation with respect to the invariant measure on the 
right hand side of (\ref{explam}) would drop out by the integration by parts. 
The problem is, however, the lack of smoothness of those functions 
at $\,{\bm\rho}={0}$ and a more subtle
argument is required. 

\subsection{$d=1\,$  case}

\noindent In one dimension,  Eq.\,(\ref{lambdaf0}) reduces to
the identity
\qq
\lambda\ =\ \int\limits_{-\infty}^\infty
\Big(\frac{x}{\tau}\,+\,\frac{1}{2}L\ln{(1+x^2)}\Big)
\eta(x)\,dx
\qqq
with $\,\eta(x)\,$ given by Eq.\,(\ref{mu1}). \,Since the latter integral
represents the integration of a smooth function against a smooth measure
on $\,S^1$, \,it converges absolutely. Consequently, the formula for 
$\,\lambda\,$ may be rewritten in the form:  
\qq
\lambda\ =\ \lim\limits_{n\to\infty}\int\limits_{-n}^n
\Big(\frac{x}{\tau}\,+\,\frac{1}{2}L\ln{(1+x^2)}\Big)\eta(x)\,dx\,.
\qqq
Now the integration by parts and the formula $\,L^\dagger\eta=0\,$
for the formal adjoint $\,L^\dagger\,$ defined with respect to
the Lebesgue measure $\,dx\,$ show that the term with 
$\,L\ln{\sqrt{1+x^2}}\,$ drops out (for the cancellation
of the boundary terms it is crucial that the integral 
is over a symmetric finite interval $\,[-n,n]$). \,We obtain this 
way the identity
\qq
\lambda\ =\ \frac{1}{\tau}\,\lim\limits_{n\to\infty}\int\limits_{-n}^nx\,\eta(x)\,dx\ 
\equiv
\ \frac{1}{\tau}\ p.v.\int\limits_{-\infty}^\infty x\,\eta(x)\,dx\,
\qqq
where ``p.v.'' stands for  ``principal value''.
\,The result may be expressed \cite{LifGredPas} by the Airy functions 
\cite{GradRyz}:
\qq
\lambda\,=\,-\frac{1}{2\tau}+\frac{1}{4\tau c^{\frac{1}{2}}}
\,\frac{d}{dc}\,\ln\big(
{\rm Ai}^2(c)+{\rm Bi}^2(c)\big)\quad\ \,{\rm for}\quad\, 
c\,=\,\frac{1}{(4\tau(A+2B))^\frac{2}{3}}\,.\ 
\qqq
The number $\,\lambda+\frac{1}{2\tau}\,$ is the Lyapunov exponent for
the one-dimensional Anderson problem (\ref{local1d}), \,recall  
relation (\ref{psirho}). \,It is always positive 
reflecting the permanent localization in one dimension. On the other 
hand, $\,\lambda\,$ 
itself changes sign as a function of $\,\tau\,$ and $\,A+2B\,$
signaling a phase transition in the one-dimensional advection 
of inertial particles \cite{WM1}. 

\subsection{$d\geq2\,$ case}

In two or more dimensions, \,Eq.\,(\ref{lambdaf0}) becomes
\qq
\lambda\ =\ \int
\Big(\frac{x}{\tau}\,+\,\frac{1}{2}L\ln{(1+x^2+y^2)}\Big)\eta(x,y)\,dx\,dy\,,
\qqq
where $\,\eta(x,y)\,$ is the density of the invariant measure
from Eqs.\,(\ref{2dinvm}) or (\ref{3dinvm}). \,The asymptotic 
behavior of $\,\eta(x,y)\,$ was established in Sec.\,\ref{subsec:ergod2}
and Sec.\,\ref{subsec:ergod3}. We show in Appendix \ref{app:B} that
it guarantees that the term with $\,L\,\ln{(1+x^2+y^2)}\,$ may, indeed, 
be dropped from the expectation on the right hand side of 
Eq.\,(\ref{explam}) so that
\qq
\lambda\ =\ \frac{1}{\tau}\int x\,\eta(x,y)\,dx\,dy\,,
\label{lyap23}
\qqq
where the integral converges absolutely as follows from the
estimates (\ref{bound2}) and (\ref{bound3}). 
\,In general, there is no closed
analytic expressions for the right hand side, unlike in the one-dimensional
case. The results of numerical simulations for $\,\lambda$, \,indicating
its qualitative dependence on the parameters of the model, \,together 
with analytic arguments about its behavior when $\,A\tau\to\infty\,$ with 
$\,A/B=const.\,$ or when $\,\tau\to 0\,$ with $\,A/\tau=const.\,$ 
and $\,B/\tau=const.\,$ may be found in \cite{MW3,MWDWL,Horvai,BCH1,BCHT}.

\nsection{\bf CONCLUSIONS}

\noindent We have studied rigorously a simple stochastic differential
equation (SDE) used to model the pair dispersion of close inertial 
particles moving in a moderately turbulent flow 
\cite{Piterbarg,WM1,MW3,MWDWL,Horvai,BCH1}. 
We have established the smoothness of the transition probabilities
and the irreducibility of the dispersion process using 
H\"ormander criteria for hypoellipticity and control theory.
For the projectivized version of the dispersion process, these results 
implied the existence of the unique invariant probability measure 
with smooth positive density as well as exponential mixing. The latter 
properties permitted to substantiate the formulae for the top Lyapunov 
exponents for the inertial particles used in the physical literature. 
In two space dimensions, we also showed that the complex-projectivized
version of the dispersion process is non-explosive when described in the
inhomogeneous variable of the complex projective space, unlike the 
real-projective version of the dispersion in one space dimension. 
A similar result was established in $\,d\geq3\,$ for the complex-valued process
built from the $\,SO(d)\,$ invariants of the projectivized dispersion
that was shown to stay for all times in the upper half-plane.
These non-explosive behaviors are the reason why the numerical simulation
of the processes in the complex (half-)plane could lead to reliable numerical  
results \cite{BCH1}. There are other questions about the models studied 
here that may 
be amenable to rigorous analysis. Let us list some of them: What about 
the expressions for the other $\,2d-1\,$ Lyapunov exponents 
(the $\,2d\,$ exponents have to sum to $\,-d/\tau$ \cite{FouHorv2})? 
Can one establish the existence of the large deviation regime for 
the finite-time Lyapunov exponents (the corresponding rate function 
for the top exponent was numerically studied in $\,d=2\,$ model 
in \cite{BCH1,BCHT}; it gives access to more subtle information about 
the clustering of inertial particles than the top Lyapunov exponent 
itself)? Is the SDE modeling the inertial particle dispersion in fully 
developed turbulence, that was introduced and studied numerically in 
\cite{BCH2,BCHT}, amenable to rigorous analysis? All those open problems 
are left for a future study.

% begin appendix
\appendix

% add ``Appendix'' to the section heading
\newcommand{\appsection}[1]{\let\oldthesection\thesection
  \renewcommand{\thesection}{%Appendix 
{\bf\oldthesection}}
  \section{#1}\let\thesection\oldthesection\setcounter{equation}{0}}

\ 

%\appendix{A}
\appsection{}
\label{app:A}

\noindent We establish here the expressions (\ref{dmu0}) and (\ref{eta0})
for the $\,SO(2d)\,$-invariant normalized volume measure on $\,S^{2d-1}\,$
for $\,d\geq 3\,$ (the same proof works also in $\,d=2$, \,although
there, the corresponding formulae are straightforward and well known). 
\,Note that for $\ S^{2d-1}\,$ identified with the set of 
$\,({\bm\rho},{\bm\chi})\in\mathbb{R}^{2d}\,$ such that
$\,{\bm\rho}^2+{\bm\chi}^2=R^2\,$ for fixed $\,R$, \,we may write
\qq
d\mu_0\,=\,const.\ \delta(R-\sqrt{{\bm\rho}^2+{\bm\chi}^2})
\,d{\bm\rho}\,d{\bm\chi}\,,
\label{mu0h}
\qqq
provided that we identify functions on $\,S^{2d-1}\,$ with homogeneous function 
of degree zero on $\,\mathbb{R}^{2d}\setminus\{0\}$.
\,Let us parametrize:
\qq
{\bm\rho}=O(\rho,0,\dots,0),\qquad {\bm\chi}=O(\rho x,\rho y,0,\dots,0)
\label{param}
\qqq
where $\,\rho=|{\bm\rho}|$, $\,\chi=|{\bm\chi}|$, $\,x\,$ and $\,y\,$
are the $\,SO(d)$-invariants of Eq.\,(\ref{invs3}), and 
$\,O\in SO(d)$. \ Note that $\,O^{-1}\,$ is the rotation that
aligns $\,{\bm\rho}\,$ with first positive half-axis of $\,\mathbb{R}^d\,$
and brings $\,{\bm\chi\,}$ into the 
half-plane spanned by the first axis and the second positive half-axis,
as required in Sec.\,\ref{subsec:ergod3}. $\,O\,$ and $\,OO'$, \,for 
$\,O'\,$ rotating in the subspace orthogonal to the first two axes, 
\,give the same $\,({\bm\rho},{\bm\chi})$. Let 
$\,\Lambda\,$ be a $\,d\times\,d$ antisymmetric matrix, $\,\Lambda_{ij}
=-\Lambda_{ji}$. \,Setting $\,O=\ee^\Lambda\,$ and differentiating
Eqs.\,(\ref{param}) at $\,\Lambda=0$, \,we obtain:
\qq
d{\bm\rho}&=&(d\rho,\rho\hspace{0.03cm}d\Lambda_{21},\dots,\rho
\hspace{0.03cm}d\Lambda_{d1})\,,\\ \cr
d{\bm\chi}&=&(-\rho y\hspace{0.03cm}d\Lambda_{21},\rho x\hspace{0.03cm}
d\Lambda_{21},\rho x\hspace{0.03cm}d\Lambda_{31}+
\rho y\hspace{0.03cm}d\Lambda_{32},\dots,\rho x\hspace{0.03cm}d\Lambda_{d1}
+\rho y\hspace{0.03cm}d\Lambda_{d2})\cr\cr
&+&(x\hspace{0.03cm}d\rho+\rho\hspace{0.03cm}dx,y\hspace{0.03cm}d\rho
+\rho\hspace{0.03cm}dy,0,\dots,0)\,. 
\qqq
Hence for the volume element, 
\qq
d{\bm\rho}\,d{\bm\chi}\ =\ \rho^{2d-1}y^{d-2}d\rho\,dx\,dy\, d\Lambda_{21}\cdots
d\Lambda_{d1}d\Lambda_{32}\cdots d\Lambda_{d2}\,.
\qqq 
The product of $\,d\Lambda_{ij}\,$ gives, modulo normalization, 
the $\,SO(d)$-invariant volume element $\,d[O]\,$ of the homogeneous 
space $\,SO(d)/SO(d-2)\,$ at point $[1]$.
\,Using the $\,SO(d)$-invariance, we infer that
\qq
d{\bm\rho}\,d{\bm\chi}\ =\ const.\ \rho^{2d-1}y^{d-2}d\rho\,dx\,dy\,d[O]\,.       
\qqq
Substituting the last expression to Eq.\,(\ref{mu0h}) and 
performing the integral
\qq
&\displaystyle{\int\delta(R-\sqrt{{\bm\rho}^2+{\bm\chi}^2})\,
\rho^{2d-1}d\rho\,=\, 
\int\delta(R-\rho\sqrt{1+x^2+y^2})\,\rho^{2d-1}d\rho}&\cr\cr
&\displaystyle{=\,\frac{R^{2d-1}}
{(1+x^2+y^2)^d}}& 
\qqq
that collects the entire $\,\rho$-dependence in the integration against
homogeneous function of zero degree,\, we obtain Eq.\,(\ref{eta0}),
modulo a constant factor that is fixed by normalizing
of the resulting measure.

\ 
 
%\appendix{B}
\appsection{} 
\label{app:B}

\noindent We show here that
\qq
\int (Lg)(x,y)\,\eta(x,y)\,dx\,dy\ =\ 0
\label{tbpr}
\qqq
in two or more dimensions, \,where $\,g(x,y)=\ln(1+x^2+y^2)\,$ and
$\,\eta(x,y)\,$ is the density of the invariant measure as defined 
by Eqs.\,(\ref{2dinvm}) and (\ref{3dinvm}). As mentioned in 
Sec.\,\ref{sec:ergod}, the identity (\ref{tbpr})
does not follow immediately by integration by parts since
the function $\,g(x,y),$ is not smooth on $\,S^{2d-1}$. 
\,We shall then replace $\,g(x,y)\,$ by the functions
\qq
g_\epsilon(x,y)\ =\ \ln\Big(1+\frac{x^2+y^2}{1+\epsilon(x^2+y^2)}\Big)\ 
=\ \ln\Big(\frac{{\bm\rho}^2+{\bm\chi}^2}{{\bm\rho}^2
+\epsilon{\bm\chi}^2}\Big)
\qqq
that are smooth on $\,S^{2d-1}\,$ for $\,\epsilon>0$.
\,The identity (\ref{tbpr}) will follow if we show that
\qq
\int(Lg)(x,y)\,\eta(x,y)\,dx\,dy\ =\ \lim\limits_{\epsilon\searrow0}\ 
\int (Lg_\epsilon)(x,y)\,\eta(x,y)\,dx\,dy\,.
\label{wills}
\qqq
Note that 
\qq
\partial_x g_\epsilon(x,y)&=&\frac{2x}{(1+(1+\epsilon)(x^2+y^2))
(1+\epsilon(x^2+y^2))}\,,\cr
\partial_y g_\epsilon(x,y)&=&\frac{2y}{(1+(1+\epsilon)(x^2+y^2))
(1+\epsilon(x^2+y^2))}\,,\cr
\partial_x^2 g_\epsilon(x,y)&=&2\frac{1+(1+2\epsilon)(y^2-x^2)
+\epsilon(1+\epsilon)(y^2-3x^2)}{(1+(1+\epsilon)(x^2+y^2))^2
(1+\epsilon(x^2+y^2))^2}\,,\cr
\partial_y^2 g_\epsilon(x,y)&=&2\frac{1+(1+2\epsilon)(x^2-y^2)
+\epsilon(1+\epsilon)(x^2-3y^2)}{(1+(1+\epsilon)(x^2+y^2))^2
(1+\epsilon(x^2+y^2))^2}
\qqq
so that
\qq
&&\hbox to 6cm{$\displaystyle{|\partial_x g_\epsilon(x,y)|\ \leq\ 
\frac{2|x|}{1+x^2+y^2}\,,}$\hfill}
|\partial_y g_\epsilon(x,y)|\ \leq\ 
\frac{2|y|}{1+x^2+y^2}\,,\cr
&&\hbox to 6cm{$\displaystyle{|\partial^2_x g_\epsilon(x,y)|\ \leq\ 10\,,}$
\hfill}|\partial^2_y g_\epsilon(x,y)|\ \leq\ 10\,.
\qqq
Using the explicit forms (\ref{Lonf2}) and (\ref{Lonf3}) of the generator
$\,L$, \,we infer that 
\qq
|(Lg_\epsilon)(x,y)|\ \leq\ C(1+|x|)
\qqq
with an $\epsilon$-independent constant $\,C$. \,Since the integral
\qq
\int(1+|x|)\,\eta(x,y)\,dx\,dy
\qqq
converges due to the estimates (\ref{bound2}) and (\ref{bound3}), 
and point-wise
\qq
\lim\limits_{\epsilon\searrow0}\ (Lg_\epsilon)(x,y)\ =\ (Lg)(x,y)\,,
\qqq
relation (\ref{wills}) follows from the Dominated Convergence
Theorem.

\

\end{document}